\documentclass[10pt,a4paper]{article}
\usepackage{graphicx}
\usepackage{amssymb,amsmath,amsthm,mathrsfs}
\usepackage{color}
\usepackage{eucal}
\numberwithin{equation}{section}
\newcommand{\ie}{\emph{i.e.}}
\newcommand{\eg}{\emph{e.g.}}
\newcommand{\cf}{\emph{cf}}

\newcommand{\D}{\mathrm{d}}
\newcommand{\C}{\mathbb{C}}
\newcommand{\N}{\mathbb{N}}
\newcommand{\R}{\mathbb{R}}
\newcommand{\Real}{\mathbb{R}}
\newcommand{\Com}{\mathbb{C}}

\newcommand{\Dom}{\mathsf{D}}
\newcommand{\Ran}{\mathsf{R}}

\newcommand{\eps}{\varepsilon}
\newcommand{\sii}{L^2}
\newtheorem{claim}{Claim}[section]
\newtheorem{theorem}[claim]{Theorem}
\newtheorem{lemma}[claim]{Lemma}

\newtheorem{proposition}[claim]{Proposition}
\newtheorem{corollary}[claim]{Corollary}
\theoremstyle{definition}
\newtheorem{remark}[claim]{Remark}

\usepackage[normalem]{ulem}
\definecolor{DarkGreen}{rgb}{0,0.5,0.1}

\newcommand\soutD{\bgroup\markoverwith
{\textcolor{DarkGreen}{\rule[.5ex]{2pt}{1pt}}}\ULon}
\newcommand\soutB{\bgroup\markoverwith
{\textcolor{blue}{\rule[.5ex]{2pt}{1pt}}}\ULon}
\newcommand{\Hm}[1]{\leavevmode{\marginpar{\tiny%
$\hbox to 0mm{\hspace*{-0.5mm}$\leftarrow$\hss}%
\vcenter{\vrule depth 0.1mm height 0.1mm width \the\marginparwidth}%
\hbox to 0mm{\hss$\rightarrow$\hspace*{-0.5mm}}$\\\relax\raggedright
#1}}}

\begin{document}

\title{\bf Asymptotic spectral analysis in colliding leaky quantum layers}
\author{Sylwia Kondej$\,^a$ \ and \ David Krej\v{c}i\v{r}\'{\i}k$\,^b$}
\date{
\small \emph{
\begin{quote}
\begin{itemize}
\item[$a)$]
Institute of Physics, University of Zielona G\'ora, ul.\ Szafrana
4a, 65246 Zielona G\'ora, Poland; s.kondej@if.uz.zgora.pl
\item[$b)$]
Department of Theoretical Physics, Nuclear Physics Institute ASCR,
25068 \v{R}e\v{z}, Czech Republic; krejcirik@ujf.cas.cz
\end{itemize}
\end{quote}
}
\medskip
23 June 2016}
\maketitle

\begin{abstract}
\noindent
We consider the Schr\"{o}dinger operator with a complex delta interaction
supported by two parallel hypersurfaces in the Euclidean space
of any dimension.
We analyse spectral properties of the system in the limit
when the distance between the hypersurfaces tends to zero.
We establish the norm-resolvent convergence to a limiting operator
and derive first-order corrections for the corresponding eigenvalues.
%
%
\end{abstract}
%

\section{Introduction}\label{Sec.intro}
%
Semiconductor heterostructures have had tremendous impact
on science and technology as building blocks
for a bottom-up approach to the fabrication
of nanoscale devices.
A key property of these material systems
is the unique versatility in terms of geometrical
dimensions and composition
and their ability to exhibit quantum effects.
Theoretical studies have lead to interesting mathematical
problems which involve an interaction of differential geometry,
spectral analysis and theory of partial differential equations.
In this paper, we rely on the mathematical concept
of \emph{leaky quantum graphs} or \emph{waveguides}
introduced by Exner and Ichinose in 2001~\cite{EI}
(see~\cite{Exner_2008} for a survey),
where the quantum Hamiltonian is modelled
by the Schr\"odinger operator with a Dirac-measure potential
supported on a hypersurface in~$\Real^d$.

The situations $d=1,2,3$ are of particular interest
in the context of mesoscopic physics of nanostructures,
where they are sometimes referred to as
\emph{quantum dots}, \emph{wires} or \emph{layers}, respectively.
We adopt the last terminology
to emphasise the geometric complexity of the problem,
but any value $d \geq 1$ is allowed in this paper.
Using the Dirac-measure interaction
instead of a regular potential
to describe a quantum particle in a nanostructure
is a simplification
in the sense that the former vanishes outside the hypersurface.
At the same time, it is a more realistic model
than considering the particle confined to
a tubular neighbourhood of the hypersurface
by means of Dirichlet boundary conditions
(see~\cite{DEK2}, \cite{CEK}, \cite{LL1}, \cite{Lu-Rowlett_2012} and~\cite{KL}
for this type of models in the case $d=3$),
because it takes into account \emph{tunnelling},
property which is observed and measured in realistic heterostructures
(see, \eg, \cite{Bjork_2002} and~\cite{Franceschi_2003}).

The objective of this paper is to quantify the effect of tunnelling
by considering coalescing heterostructures modelled by Dirac-measure potentials
imposed on two parallel hypersurfaces separated by a distance~$\eps$
and studying spectral properties in the limit as~$\eps$ tends to zero.
Spectral asymptotics of systems with leaky quantum waveguides
have been analysed in various contexts and dimensions recently
(see, \eg,
\cite{Behrndt-Exner-Lotoreichik_2014},
\cite{Behrndt-Grubb-Langer-Lotoreichik_2015},
\cite{Behrndt-Langer-Lotoreichik_2013},
\cite{Duchene-Raymond_2014},
\cite{Exner-Pankrashkin_2014},
\cite{Kondej-Lotoreichik_2014},
\cite{Lotoreichik-Ourmieres_2016}).
The geometric setting introduced in this paper is new and interesting
both physically and mathematically.
In fact, to establish the eigenvalue asymptotics as $\eps \to 0$,
we need to combine diverse methods of Riemannian geometry,
spectral analysis and theory of partial differential equations.

Motivated by a growing interest in non-self-adjoint operators
in recent years
(\cf~the review article~\cite{KSTV}
and the book chapter~\cite{KS-book} and references therein),
in this paper we proceed in a great generality
by allowing \emph{complex} couplings on the colliding hypersurfaces.
In quantum mechanics, non-self-adjoint operators are traditionally
relevant as effective models of open systems
and, more recently, as an unconventional representation of physical observables.
Schr\"odinger operators with complex delta interactions
are specifically used in Bose-Einstein condensates,
where the imaginary part of the complex coupling
models the injection and removal of particles
(see~\cite{Cartarius-Haag-Dast-Wunner_2012} and~\cite{Dohnal-Siegl}).

Let us now specify the mathematical model of this paper
and present our main results.
Let~$\Omega$ be a bounded smooth open set in~$\Real^d$ with $d \geq 1$
and let us denote by $\Sigma_0 := \partial\Omega$ the boundary of~$\Omega$.
For all sufficiently small positive~$\eps$,
we consider parallel hypersurfaces
\begin{equation}\label{hypersurface}
  \Sigma_{\pm\eps} := \{
  q \pm \eps n(q) : q \in \Sigma_0 \}
  \,,
\end{equation}
where $n : \Sigma_0 \to \Real^d$ denotes the outer unit normal to~$\Omega$.
Finally, given two constants $\alpha_\pm \in \Com$,
we consider the operator in $\sii(\Real^d)$
represented by the formal expression
\begin{equation}\label{Hamiltonian}
  H_\eps := -\Delta
  + \alpha_+ \, \delta_{\Sigma_{+\eps}}
  + \alpha_- \, \delta_{\Sigma_{-\eps}}
  \,,
\end{equation}
where~$\delta_\Sigma$ denotes the Dirac delta function
supported by a hypersurface $\Sigma \subset \Real^d$.
The purpose of this paper is to study spectral
properties of~$H_\eps$ in the limit when $\eps \to 0$.

First of all,
it is natural to expect that the limiting operator is given by
\begin{equation}\label{Hamiltonian0}
  H_0 := -\Delta + (\alpha_+ + \alpha_-) \, \delta_{\Sigma_{0}}
  \,.
\end{equation}
In this paper, we show that the convergence holds
in the norm-resolvent sense.
\begin{theorem}\label{Thm.nrs}
For any $z \in \rho(H_0)$, there exists a positive constant~$\eps_0$
such that, for all $\eps < \eps_0$,
we have $z \in \rho(H_\eps)$ and
\begin{equation}\label{nrs}
  \left\|
  (H_\eps-z)^{-1} - (H_0-z)^{-1}
  \right\|_{\sii(\Real^d) \to \sii(\Real^d)}
  = O(\eps)
  \qquad \mbox{as} \qquad \eps \to 0
  \,.
\end{equation}
\end{theorem}

As a consequence of Theorem~\ref{Thm.nrs},
we obtain a convergence of the spectrum of~$H_\eps$
to the spectrum of~$H_0$ as $\eps \to 0$.
In particular, discrete eigenvalues change continuously with~$\eps$
(\cf~\cite[Sec.~IV.3.5]{Kato}).
By a discrete eigenvalue~$\lambda_\eps$ of~$H_\eps$
we mean an isolated eigenvalue of finite algebraic multiplicity
such that the range of $H_\eps - \lambda_\eps$ is closed.
We remark that~$H_0$ may or may not possess discrete eigenvalues,
depending on the values of the coupling constants~$\alpha_\pm$
and geometry of~$\Sigma_0$;
in particular, they always exist in the self-adjoint case
if the constants are negative and sufficiently large.
Since the interaction in~\eqref{Hamiltonian}
is compactly supported in~$\Real^d$,
it is also possible to show that the essential spectrum of~$H_\eps$
(\ie~the complement of the discrete eigenvalues in the spectrum)
equals the essential spectrum of the self-adjoint  Laplacian without
the delta interactions, \ie
$$
  \sigma_\mathrm{ess}(H_\eps) = [0,+\infty)
  \,,
$$
for all $\eps \geq 0$,
regardless of the geometry of~$\Sigma_0$ and values of~$\alpha_\pm$.

The main interest of Theorem~\ref{Thm.nrs}
lies in the sharpness of the power of~$\eps$ in~\eqref{nrs}.
Indeed, as the next result of this paper,
we derive the following asymptotics for simple eigenvalues.
\begin{theorem}\label{Thm.evs}
Let~$\lambda_0$ be a simple discrete eigenvalue of~$H_0$
and let~$\psi_0$
be the corresponding eigenfunction.
There exist positive constants~$\eps_0$ and~$r$ such that,
for all $\eps < \eps_0$, $H_\eps$~possesses precisely one discrete eigenvalue
of algebraic multiplicity one
in the open ball~$B_r(\lambda_0)$ disk of radius~$r$ centred at~$\lambda_0$.
Moreover, the following asymptotics holds:
\begin{equation}\label{as}
  \lambda_\eps = \lambda_0 + \lambda_0' \, \eps + O(\eps^2)
  \qquad \mbox{as} \qquad
  \eps \to 0
\end{equation}
with
\begin{equation}\label{as.bis}
  \lambda_0' := \frac{\displaystyle
  \alpha_+
  \int_{\Sigma_0} \partial_n^+ \psi_0^2
  + \alpha_-   \int_{\Sigma_0} \partial_n^- \psi_0^2
  - \int_{\Sigma_0}
  \left[\alpha_+^2+\alpha_-^2 + (\alpha_+ -\alpha_-) \,(d-1) K_1 \right]
  \psi_0^2}
  {\displaystyle \int_{\Real^d} \psi_0^2}
  \,,
\end{equation}
where~$K_1$ denotes the first mean curvature of~$\Sigma_0$ and
$$
  \partial^\pm_n f(x)
  :=
  \lim_{\epsilon \to 0^+} \frac{f(x\pm n\epsilon )-f(x)}{\epsilon}
  \,.
$$
\end{theorem}

The functions appearing in the numerator of~\eqref{as.bis}
should be understood in the sense of traces
and their rigorous definition will be provided in the following section.

 We also give an analogous theorem
for degenerate semisimple eigenvalues,
\ie\ for the case when
the algebraic and geometric multiplicity coincide, (\cf~\cite[Sec.~I.5.3]{Kato}).
This result is formulated as Theorem~\ref{th-deg} below.

We remark that a presence of the first mean curvature
in eigenvalue asymptotics has been recently observed
in related problems,
see~\cite{K5}, \cite{K10} and~\cite{Pankrashkin-Popoff_2015b}.

If $\alpha_+ = \alpha_-$, formula~\eqref{as.bis}
simplifies to $\lambda_0' = 2 \alpha_+^2$ (\cf~\eqref{conditions}),
so the first correction term in the eigenvalue asymptotics
is insensitive to the geometric setting if the coupling constants coincide.

We stress that the asymptotics~\eqref{as}
is not a consequence of analytic perturbation theory.
As a matter of fact, taking a formal derivative of~$\lambda_\eps$
with respect to~$\eps$ in the spirit of the Hellmann-Feynman theorem
would lead only to the first integral in the numerator of~\eqref{as.bis}.
Of course, this formal manipulation is not justified
because of the singular dependence of~$H_\eps$ on~$\eps$.
It is interesting that a non-trivial rigorous approach
is needed to reveal the geometric term in~\eqref{as.bis}.

This paper is organised as follows.
In Section~\ref{Sec.pre} we present some necessary
analytic and geometric prerequisites.
The norm-resolvent convergence of Theorem~\ref{Thm.nrs}
is established in Section~\ref{Sec.nrs}.
Our strategy is to derive first estimates for the norm
of the resolvent as an operator between Sobolev spaces,
which we believe are of interest on its own.
In Section~\ref{Sec.efs}, we establish a uniform
convergence of eigenfunctions by a refined application
of the maximum principle.
Section~\ref{Sec.evs} is devoted to a proof of Theorem~\ref{Thm.evs}
based on a detailed study of spectral projections
as well as to its extension to degenerate eigenvalues.
We conclude the paper by an appendix (Section~\ref{Sec.warm}),
where Theorem~\ref{Thm.evs} is re-established in the simplest case $d=1$.
Here the eigenvalue problem can be reduced to
a transcendental equation,
for which the implicit function theorem yields
the the first correction term.

\section{Preliminaries}\label{Sec.pre}
%
Let us start by properly defining the operators~$H_\eps$ and~$H_0$
(the latter can be considered as~$H_\eps$ for $\eps=0$
if we set $\Sigma_{\pm 0} := \Sigma_0$).
The sum in~\eqref{Hamiltonian} has a good meaning
as a sum of bounded operators from the Sobolev space $H^{1}(\Real^d)$
to its dual $H^{-1}(\Real^d)$.
It is more customary to consider~$H_\eps$ as an unbounded
operator in the Hilbert space $\sii(\Real^d)$.
To this purpose, we introduce the quadratic form
\begin{equation}\label{form}
  h_\eps[\psi] := \int_{\Real^d} |\nabla\psi|^2
  + \alpha_+ \int_{\Sigma_{+\eps}} |\psi|^2
  + \alpha_- \int_{\Sigma_{-\eps}} |\psi|^2
  \,, \qquad
  \Dom(h_\eps) := H^{1}(\Real^d)
  \,,
\end{equation}
which is formally associated with the expression
on the right hand side of~\eqref{Hamiltonian},
and define~$H_\eps$ as the unique m-sectorial operator
associated with~$h_\eps$ via the first representation theorem
(\cf~\cite[Thm.~VI.2.1]{Kato}).

The boundary terms in~\eqref{form} should be understood
in the sense of traces (\cf~\cite{Adams}).
More specifically, in analogy with~\eqref{hypersurface},
we introduce a mapping
\begin{equation}\label{tube}
  \mathcal{L}: \Sigma_0 \times \Real \to \Real^d :
  \left\{ (q,t) \mapsto q + t \, n(q)  \right\}
\end{equation}
and define sets $\Sigma_t := \mathcal{L}(\Sigma_0 \times \{t\})$.
Because of the boundedness and smoothness of~$\Omega$,
there exists a positive number~$a$
such that
\begin{equation}\label{diffeomorphism}
  \mbox{$\mathcal{L} :\Sigma_0 \times [-a,a] \to
  \mathcal{L}(\Sigma_0 \times [-a,a])$
  is a diffeomorphism.}
\end{equation}
Consequently, $\Sigma_t$~is a smooth hypersurface
(parallel to~$\Sigma_0$ at distance~$|t|$)
for all $|t| \leq a$.
(Neither~$\Sigma_0$ nor $\Sigma_{t}$ are necessarily connected.)
By the trace embedding theorem (\cf~\cite[Thm.~5.36]{Adams}),
the \emph{trace operator}
\begin{equation}\label{trace}
  \tau_t : H^1(\Real^d) \to \sii(\Sigma_t)
\end{equation}
is bounded for all $|t| \leq a$.
In fact, if $|t| \leq a$,
then for any $\delta>0$ there exists a positive constant~$C_\delta$
(depending in addition to~$\delta$ also on the geometry of~$\Sigma_0$)
such that,
for all $\psi \in H^1(\Real^d)$,
\begin{equation}\label{trace.norm}
  \|\tau_t\psi\|_{\sii(\Sigma_t)}^2
  \leq \delta \, \|\nabla\psi\|_{\sii(\Real^d)}^2
  + C_\delta \, \|\psi\|_{\sii(\Real^d)}^2
  \,.
\end{equation}
This estimate can be proved in a standard way by using
the diffeomorphism~$\mathcal{L}$ and the one-dimensional bound
\begin{equation}\label{elementary}
  \sup_{(-l,l)}|\varphi|^2
  \leq 2 \, \|\varphi\|_{\sii(-l,l)} \, \|\varphi'\|_{\sii(-l,l)}
  +  \, (2l)^{-1} \, \|\varphi\|_{\sii(-l,l)}^2
\end{equation}
valid for all $\varphi \in H^1((-l,l))$,
where~$l$ is any positive number.
It follows that the boundary terms in~\eqref{form}
(in which we ambiguously write~$\psi$ instead of~$\tau_{\pm\eps} \psi$)
represent a relatively bounded perturbation
of the gradient integral with the relative bound equal to zero
(since~$\delta$ can be taken arbitrarily small).
Consequently, the form~\eqref{form} is closed
by classical perturbation results (\cf~\cite[Thm.~VI.1.33]{Kato}),
so that the first representation theorem applies.

Next we set
\begin{equation}\label{Omega.sets}
\begin{aligned}
  \Omega_{\eps}^{0}
  &:= \left\{
  \mathcal{L}(q,t) : \, q \in \Sigma_0, \, -\eps < t < +\eps
  \right\}
  ,
  \\
  \Omega_{\eps}^{\pm}
  &:= \left\{
  \mathcal{L}(q,t) : \, q \in \Sigma_0, \, \eps < \pm t < a/2
  \right\}
  ,
\end{aligned}
\end{equation}
where $0 \leq \eps < a/2$
($\Omega_{0}^{0}$ is an empty set).
In words, $\Omega_{\eps}^{0}$ (respectively, $\Omega_{\eps}^{\pm}$)
is the open set squeezed between
the parallel hypersurfaces~$\Sigma_{+\eps}$ and~$\Sigma_{-\eps}$
(respectively, $\Sigma_{\pm\eps}$ and~$\Sigma_{\pm a/2}$).
For positive~$\eps$, the trace operators
\begin{equation}\label{traces}
\begin{aligned}
  \tau_{-\eps}^- &: H^2(\Omega_{\eps}^{-}) \to H^1(\Sigma_{-\eps})
  \,,
  &\tau_{+\eps}^+ &: H^2(\Omega_{\eps}^{+}) \to H^1(\Sigma_{+\eps})
  \,,
  \\
  \tau_{-\eps}^+ &: H^2(\Omega_{\eps}^{0}) \to H^1(\Sigma_{-\eps})
  \,,
  &\tau_{+\eps}^- &: H^2(\Omega_{\eps}^{0}) \to H^1(\Sigma_{+\eps})
  \,,
\end{aligned}
\end{equation}
are again bounded by the trace embedding theorem.
The claim applies to the first line even if $\eps=0$.
By using the first representation theorem
and elliptic regularity theory,
it is standard to show that~$H_\eps$ acts as the Laplacian,
subject to the interface conditions
\begin{equation}\label{conditions}
\left\{
\begin{aligned}
  \tau_{\pm\eps}^+ \partial_n \psi
  - \tau_{\pm\eps}^- \partial_n \psi
  &= \alpha_\pm \tau_{\pm\eps} \psi
  && \mbox{on} \quad
  \Sigma_{\pm \eps}
  & \mbox{if}\quad \eps>0
  \,,
  \\
  \tau_{+0}^+ \partial_n \psi
  - \tau_{-0}^- \partial_n \psi
  &= (\alpha_+ + \alpha_-) \;\! \tau_{0} \psi
  && \mbox{on} \quad
  \Sigma_{0}
  & \mbox{if}\quad \eps=0
  \,.
\end{aligned}
\right.
\end{equation}
More precisely, we have
\begin{equation}\label{Hamiltonian.domain}
\begin{aligned}
  H_\eps\psi &= -\Delta\psi
  \qquad \mbox{a.e.\ in } \Real^d
  \,,
  \\
  \Dom(H_\eps) &= \left\{
  \psi \in H^1(\Real^d)
  \cap H^2\big(\Real^d\setminus(\Sigma_{+\eps}\cup \Sigma_{-\eps})\big)
  : \psi \mbox{ satisfies~\eqref{conditions}}
  \right\}
  \,.
\end{aligned}
\end{equation}
The meaning of
the trace maps $\partial_n^\pm \psi \in L^2 (\Sigma_0 ) $
used in formula~\eqref{as.bis}
is precisely
$
  \partial_n^\pm \psi  := \pm \tau_{\pm0}^\pm \partial_n \psi
$.

Next, we overtake from~\cite{KRT} some facts
about the geometry of parallel hypersurfaces.
In view of~\eqref{diffeomorphism},
$\Omega_a:=\mathcal{L}(\Sigma_0 \times (-a,a))$
can be identified with the Riemannian manifold
$\Sigma_0 \times (-a,a)$ equipped with
the metric~$G$ induced by~\eqref{tube};
it has a block form
\begin{equation}\label{metric}
  G(q,t) = g(q) \circ \big(I-t \, L(q)\big)^2 + \D t^2
  \,,
\end{equation}
where~$g$ is the Riemannian metric of~$\Sigma_0$,
$L := -\D n$ is the Weingarten map of~$\Sigma_0$
and~$I$ denotes the identity map on $T_q\Sigma_0$.

It follows from~\eqref{metric}
that $|G| := \det(G) = |g| \, f^2$
with $|g| := \det(g)$ and
\begin{equation}\label{Jacobian}
  f(q,t) := \prod_{\mu=1}^{d-1} \big(1-t \, \kappa_\mu(q)\big)
  = 1+\sum _{\mu=1}^{d-1} (-t)^\mu
  \begin{pmatrix} d-1 \\ \mu\end{pmatrix}
  K_\mu(q)
  \,,
\end{equation}
where $\kappa_1, \dots,\kappa_{d-1}$ are the \emph{principal curvatures}
and $K_\mu$ is the $\mu^\mathrm{th}$ \emph{mean curvature} of~$\Sigma_0$
(\cf~\cite{Kuhnel}).
Since the first mean curvature appears in Theorem~\ref{Thm.evs},
we remark that, locally,
$$
  K_1 = \frac{\kappa_1 + \dots + \kappa_{d-1}}{d-1}
  \,.
$$
The sign of~$K_1$ depends on the orientation of~$\Sigma_0$;
in our case where~$\Sigma_0$ is assumed to be oriented
via the outer normal~$n$ to~$\Omega$,
we have $K_1 \leq 0$ if~$\Omega$ is convex (\cf~\cite{K10}).
It follows from~\eqref{Jacobian} that the surface elements
of~$\Sigma_0$ and~$\Sigma_t$ are related by
\begin{equation}\label{elements}
  \D\Sigma_t  =  f(q,t) \, \D\Sigma_0
  \,,
\end{equation}
where $\D\Sigma_0 = |g(q)|^{1/2} \D q$.

From~\eqref{diffeomorphism} and~\eqref{Jacobian}, we deduce
\begin{equation}\label{Ass}
  \forall |t| \leq a
  \,, \quad t \, \max\{\|\kappa_1\|_\infty,\dots,\|\kappa_{d-1}\|_\infty\}  < 1
  \,,
\end{equation}
so that  $\inf_{q\in\Sigma_0}f(q,t) > 0$ for every~$t$ such that $|t| \leq a$.
In particular, there exists a positive constant~$C$
(depending on~$a$ and the supremum norms of the principal curvatures)
such that,
for all $(q,t) \in \Sigma_0 \times [-a,a]$,
\begin{equation}\label{f.bound}
  C^{-1} \leq f(q,t) \leq C
\end{equation}
and
\begin{equation}\label{G.bound}
  C^{-1} \, g(q) + \D t^2 \leq G(q,t) \leq C \, g(q) + \D t^2
  \,.
\end{equation}
(Hereafter we adopt the convention that~$C$
denotes a generic constant whose value can change from line to line.)

Given a coordinate system $(q,t) \in \Sigma_0 \times (-a,a)$,
we denote by~$G_{ij}$ and~$G^{ij}$ the corresponding
coefficients of~$G$ and~$G^{-1}$.
We also adopt the Einstein summation convention,
the range of Latin and Greek indices being $1,\dots,d$
and $1\dots d-1$, respectively,
and abbreviate $\partial_i := \partial/\partial q^i$
with $q^d := t$ (we shall also write $\partial_t := \partial_d$).
It will be convenient to choose for $q=(q^1,\dots,q^{d-1})$
the Riemannian normal coordinates in~$\Sigma_0$,
which exist in a neighbourhood of any point of~$\Sigma_0$.
In these coordinates, since~$\Omega$ is smooth and bounded,
there exists a positive number~$r_0$
such that, for any $p \in \Sigma_0$, the useful estimates
\begin{equation}\label{normal}
  C^{-1} \, (\delta_{\mu\nu}) \leq (g_{\mu\nu}) \leq C \, (\delta_{\mu\nu})
  \,, \qquad
  | \partial_\rho \, g_{\mu\nu} | \leq C
  \,,
\end{equation}
hold in the geodesic ball of radius~$r_0$ centred at~$p$.

Finally, we remark that the mapping~\eqref{tube}
induces a natural unitary transform
between Hilbert spaces
\begin{equation}\label{unitary}
  \mathcal{U} : \sii(\Omega_a) \to
  \sii\left(\Sigma_0\times(-a,a),
  |G(q,t)|^{1/2} \,\D q \wedge \D t\right)
  : \{ \psi \mapsto \psi \circ \mathcal{L} \}
  \,.
\end{equation}
In particular, it will enable us to relate
$\sii(\Sigma_t)$  and $\sii(\Sigma_0)$. Since~$H_\eps$ acts as the Laplacian in
$\Real^d \setminus (\Sigma_{+\eps} \cup \Sigma_{-\eps})$,
its action in the curvilinear ``coordinates'' $(q,t)$
induced by~$\mathcal{L}$ is given by Laplace-Beltrami operator
\begin{equation}\label{LB}
  -\Delta_G := -|G|^{-1/2} \partial_i |G|^{1/2} G^{ij} \partial_j
\end{equation}
 in
$
  \Sigma_0 \times \left[(-a,-\eps) \cup (-\eps,\eps) \cup (\eps,a)\right]
$.
Given $\psi \in L^2 (\R^d)$,
we shall occasionally write $\mathcal{U}\psi$,
meaning that~$\mathcal{U}$ acts on the restriction
of~$\psi$ to $L^2 (\Omega_a)$.
We point out the following topological equivalence
of Sobolev spaces.
\begin{lemma}\label{Lem.equivalence}
There exists a positive constant~$C$ such that,
for every $a \leq t_1 < t_2 \leq a$,
$$
  C^{-1} \, \|\mathcal{U}\psi\|_{H^2(\Sigma_0\times(t_1,t_2))}
  \leq \|\psi\|_{H^2(\mathcal{L}(\Sigma_0\times(t_1,t_2)))} \leq
  C \, \|\mathcal{U}\psi\|_{H^2(\Sigma_0\times(t_1,t_2))}
$$
for every $\psi \in H^2(\mathcal{L}(\Sigma_0\times(t_1,t_2)))$.
\end{lemma}
\begin{proof}
The proof is a straightforward application of~\eqref{tube},
namely~\eqref{metric} with estimates~\eqref{normal}.
We leave the details to the reader.
\end{proof}
%

\section{The norm-resolvent convergence}\label{Sec.nrs}
%
The objective of this section is to prove Theorem~\ref{Thm.nrs}.
For all $\eps \geq 0$ small enough  and $z \in \rho(H_\eps)$,
we set
$$
  R_\eps(z) := (H_\eps-z)^{-1}
  \,.
$$
Given $\Psi \in \sii(\Real^d)$, the function $\psi_\eps := R_\eps(z) \Psi$
is the unique solution of the resolvent equation $(H_\eps-z)\psi_\eps=\Psi$.
The weak formulation of the problem reads
\begin{equation}\label{weak}
  \forall \varphi \in H^1(\Real^d)
  \,, \qquad
  h_\eps(\varphi,\psi_\eps) - z \, (\varphi,\psi_\eps)_{\sii(\Real^d)}
  = (\varphi,\Psi)_{\sii(\Real^d)}
  \,,
\end{equation}
where $h_\eps(\cdot,\cdot)$ is the sesquilinear form
associated with~\eqref{form}.

First of all, we show that the resolvent~$R_\eps(z)$
is uniformly bounded as $\eps \to 0$.
\begin{lemma}\label{Lem.res.bound1}
There exist constants~$z_0 \in \R$ and~$C>0$ such that,
for all $0 \leq \eps < a$
and every $z \in \Com$ such that $\Re z < z_0$,
$z \in \rho(H_\eps)$ and
\begin{equation}\label{res.bound1}
  \left\| R_\eps(z) \right\|_{\sii(\Real^d) \to H^1(\Real^d)}
  \leq C
  \,.
\end{equation}
\end{lemma}
\begin{proof}
Since~$H_\eps$ is m-sectorial, we know that the claim
holds with an \emph{a priori} $\eps$-dependent constant~$z_0$.
The content of the lemma is that~$z_0$ can be made
actually independent of~$\eps$.
Choosing the test function $\varphi:=\psi_\eps$ in~\eqref{weak},
taking the real part of the obtained identity
and applying~\eqref{trace.norm}
together with the Schwarz inequality,
we get
\begin{multline*}
  \left[1-\delta \, (|\Re\alpha_+| + |\Re\alpha_-|) \right]
  \|\nabla\psi_\eps\|_{\sii(\Real^d)}^2
  - \left[C_\delta  \, (|\Re\alpha_+| + |\Re\alpha_-|) + \Re z\right]
  \|\psi_\eps\|_{\sii(\Real^d)}^2
  \\
  \leq \|\psi_\eps\|_{\sii(\Real^d)} \, \|\Psi\|_{\sii(\Real^d)}
  \,.
\end{multline*}
We choose~$\delta$ so small that
$1-\delta \, (|\Re\alpha_+| + |\Re\alpha_-|) \geq 1/2$.
It follows that every $z \in \Com$ such that
$
  \Re z < - C_\delta \, (|\Re\alpha_+| + |\Re\alpha_-|)
$
lies outside the closure of the numerical range of~$H_\eps$,
that is, inside the resolvent set $\rho(H_\eps)$
because~$H_\eps$ is m-sectorial.
Choosing
\begin{equation}\label{z0}
  z_0 := - \frac{1}{2} - C_\delta \, (|\Re\alpha_+| + |\Re\alpha_-|)
  \,,
\end{equation}
we arrive at~\eqref{res.bound1} with $C := \sqrt{8}$.
\end{proof}
\begin{remark}\label{Rem.duality}
It is also possible to  look for solutions of~\eqref{weak}
for $\Psi \in H^{-1}(\Real^d)$ in which case the right hand side
must be understood as the duality pairing
between $H^1(\Real^d)$ and $H^{-1}(\Real^d)$.
Proceeding as in the previous proof, with a slight modification
that the Schwarz inequality is replaced by the estimate
$$
  \left| {}_{H^{1}(\Real^d)}(\psi_\eps,\Psi)_{H^{-1}(\Real^d)} \right|
  \leq \|\psi_\eps\|_{H^1(\Real^d)} \, \|\Psi\|_{H^{-1}(\Real^d)}
  \,,
$$
we obtain
$$
  \left\| R_\eps(z) \right\|_{H^{-1}(\Real^d) \to H^1(\Real^d)}
  \leq C
$$
with $C:=2$.
\end{remark}

We shall need a resolvent estimate of the type~\eqref{res.bound1}
in a better topology.
In the case of the free Hamiltonian (\ie\ $\alpha_\pm = 0$),
we know that the resolvent is bounded in the topology
of bounded operators on $\sii(\Real^d)$ to $H^2(\Real^d)$.
It does not hold if~$\alpha_+$ or~$\alpha_-$ is non-zero,
because then the functions from the domain of~$H_\eps$
are not in $H^2(\Real^d)$, \cf~\eqref{Hamiltonian.domain}.
However, the functions belong to
$H^2(\Real^d\setminus(\Sigma_{+\eps} \cup \Sigma_{-\eps}))$
and the following uniform estimate holds.
\begin{lemma}\label{Lem.res.bound2}
There exists a positive constant~$C$ such that,
for every $z \in \Com$ satisfying $\Re z < z_0$
with~$z_0$ given by~\eqref{z0}
and for all $0 \leq \eps < a/4$,
we have
\begin{equation}\label{res.bound2}
  \left\| R_\eps(z) \right\|_{\sii(\Real^d) \to
  H^2(\Real^d\setminus(\Sigma_{+\eps} \cup \Sigma_{-\eps}))}
  \leq C
  \,.
\end{equation}
\end{lemma}
\begin{proof}
The message of the lemma is that the constant~$C$ in~\eqref{res.bound2}
can be made independent of~$\eps$,
which is not \emph{a priori} clear.
Setting $\psi_\eps := R_\eps(z) \Psi$
for every $\Psi \in\sii(\Real^d)$ as above
(recall that~$\psi_\eps$ satisfies~\eqref{weak}),
estimate~\eqref{res.bound2} is equivalent to
the simultaneous validity of the bounds
\begin{align}
  \left\| \psi_\eps \right\|_{H^2(\Omega_\eps^{0})}
  &\leq C \, \|\Psi\|_{\sii(\Real^d)}
  \,,
  \label{res.bound2.1}
  \\
  \left\| \psi_\eps \right\|_{H^2(\Omega_\eps^{\pm})}
  &\leq C \, \|\Psi\|_{\sii(\Real^d)}
  \,,
  \label{res.bound2.2}
  \\
  \left\| \psi_\eps \right\|_{H^2(\Real^d \setminus
  \overline{\Omega_{a/4}^{0}})}
  &\leq C \, \|\Psi\|_{\sii(\Real^d)}
  \,,
  \label{res.bound2.3}
\end{align}
with a constant~$C$ independent of~$\Psi$ and~$\eps$.
Here the sets~$\Omega_\eps^{\pm}$ and~$\Omega_\eps^0$
are defined in~\eqref{Omega.sets}.
Note that $\psi_\eps \in H^2(\Real^d\setminus(\Sigma_{+\eps} \cup \Sigma_{-\eps}))$
is known due to~\eqref{Hamiltonian.domain};
our aim is to establish the uniform estimates
\eqref{res.bound2.1}--\eqref{res.bound2.3}.

Estimate~\eqref{res.bound2.3} follows at once
by the interior regularity of weak solutions of
the elliptic problem
$(H_\eps-z)\psi_\eps=\Psi$ in
$
  \Omega' :=
  \Real^d \setminus \overline{\Omega_{a/4}^{0}}
$;
see, \eg, \cite[Thm.~6.3.1]{Evans}
together with Lemma~\ref{Lem.res.bound1},
recall that~$H_\eps$ acts as the Laplacian in~$\Omega'$
due to~\eqref{Hamiltonian.domain}
and notice that~$\Omega'$ is independent of~$\eps$.
The validity of~\eqref{res.bound2.1} and~\eqref{res.bound2.2}
is less obvious because of
the $\eps$-dependent interface conditions~\eqref{conditions}
and a refined boundary regularity is needed.
Let us sketch the proof of~\eqref{res.bound2.1} with $\eps > 0$.
The proof of~\eqref{res.bound2.2} with $\eps \geq 0$ is analogous.
Our approach is based on elliptic regularity;
see, \eg, \cite[Sec.~6.3]{Evans} to where we refer for more details.

Recalling~\eqref{unitary}, we set $v_\eps := \mathcal{U} \psi_\eps$
and $V := \mathcal{U}\Psi$.
Let $\eta:\Real \to [0,1]$ be a smooth cut-off function,
which is equal to~$1$ on $(-a/2,a/2)$
and to~$0$ outside $(-3a/4,3a/4)$
(we keep to denote by the same symbol~$\eta$
the function $1 \otimes \eta$ on $\Sigma_0 \times (-a,a)$).
In~\eqref{weak}, let us choose the test function~$\varphi$
in the following way
$$
  (\mathcal{U}\varphi)(q,t)
  := \eta(t)^2 \, u(q,t)
  \,,
$$
where $u \in H^1(\Sigma_0 \times (-a,a))$.
Using~\eqref{tube}, the identity~\eqref{weak}
is transferred to
\begin{multline}\label{weak.tube}
  \big(
  \partial_{i} (\eta^2 u), G^{ij} \partial_j v_\eps
  \big)_\mathcal{H}
  + \alpha_+ \int_{\Sigma_0} (\bar{u}v_\eps f)(q,\eps) \, \D\Sigma_0
  + \alpha_- \int_{\Sigma_0} (\bar{u}v_\eps f)(q,-\eps) \, \D\Sigma_0
  \\
  - z \, \big( \eta^2 u, v_\eps \big)_\mathcal{H}
  = \big(\eta^2 u,V\big)_{\mathcal{H}}
  \,,
\end{multline}
where~$\mathcal{H}$ denotes the target Hilbert space in~\eqref{unitary}.

In~\eqref{weak.tube}, we choose
\begin{equation} \label{eq-defu}
  u(q,t) := - \partial_\rho^{-h} \partial_\rho^h v_\eps(q,t)
  \,,
\end{equation}
where $\rho \in \{1,\dots,d-1\}$
and $\partial_\rho^h v_\eps(q,t)$ is
the $\rho^\mathrm{th}$ difference quotient of size $h>0$
(\cf~\cite[Sec.~5.8.2]{Evans})
$$%
\partial_\rho^h v_\eps(q,t)
  := \frac{v_\eps(q^1,\dots, q^\alpha + h, \dots q^{d-1},t)-v_\eps(q,t)}{h}
  \,.
$$
Using the ``integration-by-parts'' rule
for the difference quotients
(\cf~\cite[proof of Thm.~5.8.3]{Evans})
and sending~$h$ to zero,
we get
\begin{multline}\label{after0}
  \big(
  \partial_{i} \partial_\rho v_\eps,
  \eta^2 G^{ij} \partial_j \partial_\rho v_\eps
  \big)_\mathcal{H}
  + \alpha_+ \int_{\Sigma_0} |\partial_\rho v_\eps|^2 (q,\eps) \, \D\Sigma_0
  + \alpha_- \int_{\Sigma_0} |\partial_\rho v_\eps|^2(q,-\eps) \, \D\Sigma_0
  \\
  + b[v_\eps]
  + z \, \big( \eta^2 \partial_\rho^2 v_\eps, v_\eps \big)_\mathcal{H}
  = - \big(\eta^2 \partial_\rho^2 v_\eps,V\big)_{\mathcal{H}}
  \,.
\end{multline}
Here $b[v_\eps]$ is a quadratic form gathering subdominant terms
that can be treated as a perturbation of the first line in~\eqref{after0}
or integrals involving only first-order derivatives of~$v_\eps$.
Recall that, by Lemma~\ref{Lem.res.bound1} and~\eqref{G.bound},
we already know that
\begin{equation}\label{know}
  \|v_\eps\|_{H^1(\Sigma_0 \times (-a,a))}
  \leq C \, \|\Psi\|_{\sii(\Real^d)}
  \,.
\end{equation}
Writing
$$
\begin{aligned}
  \big| \big( \eta^2 \partial_\rho^2 v_\eps, v_\eps \big)_\mathcal{H} \big|
  &\leq \delta \, \|\eta \, \partial_\rho^2 v_\eps\|_{\mathcal{H}}^{2}
  + \delta^{-1} \, \|v_\eps\|_{\mathcal{H}}^{2}
  \\
  \big| \big(\eta^2 \partial_\rho^2 v_\eps,V\big)_{\mathcal{H}} \big|
  &\leq \delta \, \|\eta \, \partial_\rho^2 v_\eps\|_{\mathcal{H}}^{2}
  + \delta^{-1} \, \|V\|_{\mathcal{H}}^{2}
  \,,
\end{aligned}
$$
the first terms on the right hand side
with sufficiently small positive~$\delta$
can be treated as a perturbation of
the first dominant term of~\eqref{after0},
while we have
$
  \|V\|_{\mathcal{H}} = \|\Psi\|_{\sii(\Omega_a)}
  \leq \|\Psi\|_{\sii(\Real^d)}
$
and~\eqref{know}.
In fact, the boundary terms in~\eqref{after0} are also a perturbation
because of the following estimate based on~\eqref{elementary}:
\begin{multline*}
  \left|
  \int_{\Sigma_0} |\partial_\rho v_\eps|^2 (q,\pm\eps) \, \D\Sigma_0
  \right|
  \\
  \leq \delta \,
  \|\partial_t \partial_\rho v_\eps\|_{\sii(\Sigma_0\times(-a/2,a/2))}^2
  + (\delta^{-1}+a^{-1}) \,
  \|\partial_\rho v_\eps\|^2_{\sii(\Sigma_0\times(-a/2,a/2))}
  \,.
\end{multline*}
Summing up, from~\eqref{after0} with help of~\eqref{know}
together with~\eqref{G.bound} and~\eqref{normal},
we conclude with key estimates
\begin{equation}\label{key}
  \|\partial_i \partial_\rho v_\eps\|_{\sii(\Sigma_0\times(-a/2,a/2))}
  \leq C \, \|\Psi\|_{\sii(\Real^d)}
\end{equation}
for every $i \in \{1,\dots,d\}$ and $\rho \in \{1,\dots,d-1\}$.

To get an analogous estimate for $\partial_t^2 v_\eps$,
we employ the fact that,
by~\eqref{Hamiltonian.domain} and~\eqref{tube},
$v_\eps$~satisfies the differential equation
(recall~\eqref{LB})
\begin{equation}\label{ae}
  -|G|^{-1/2} \partial_i (|G|^{1/2} G^{ij} \partial_j v_\eps)
  - z \, v_\eps
  = V
  \qquad \mbox{a.e.\ in $\Sigma_0\times(-\eps,\eps)$}.
\end{equation}
Using the block-diagonal structure of~$G$, see~\eqref{metric},
we can cast~\eqref{ae} into the form
\begin{equation}\label{block-diagonal}
  - \partial_t^2 v_\eps = V + z\, v_\eps
  + |G|^{-1/2} \partial_\mu (|G|^{1/2} G^{\mu\nu} \partial_\nu v_\eps)
  + |G|^{-1/2} (\partial_t |G|^{1/2}) \partial_t v_\eps
  \,,
\end{equation}
where the right hand side contains no second-order
derivative of~$v_\eps$ with respect to~$t$.
Using~\eqref{key} and~\eqref{know},
we can thus conclude with the missing inequality
\begin{equation}\label{key.missing}
  \|\partial_t^2 v_\eps\|_{\sii(\Sigma_0\times(-\eps,\eps))}
  \leq C \, \|\Psi\|_{\sii(\Real^d)} \,.
\end{equation}

From~\eqref{key} and~\eqref{key.missing} together with
the first-order derivatives inequality~\eqref{know},
we have thus obtained the estimate
$
  \|v_\eps\|_{H^2(\Sigma_0\times(-\eps,\eps))}
  \leq C \, \|\Psi\|_{\sii(\Real^d)}
$.
By Lemma~\ref{Lem.equivalence},
we then get an analogous estimate for $\psi_\eps = \mathcal{U}^{-1} v_\eps$
in the Euclidean set
$\Omega_\eps^0 = \mathcal{L}(\Sigma_0\times(-\eps,\eps))$.
This concludes the sketch of the proof of~\eqref{res.bound2.1} with $\eps > 0$.
\end{proof}

Now we are in a position to prove Theorem~\ref{Thm.nrs}.
\begin{proof}[Proof of Theorem~\ref{Thm.nrs}]
Let $z \in \Com$ be such that $\Re z < z_0$,
where~$z_0$ is given by~\eqref{z0}, and $0 < \eps < a/4$.
By Lemma~\ref{Lem.res.bound1}, $z \in \rho(H_\eps)$ for all $\eps \geq 0$.
Given any $\Phi,\Psi \in \sii(\Real^d)$, we set
$\psi_\eps := R_\eps(z) \Psi$ as before
and $\phi_0 := R_0(z)^*\Phi$.
We have
\begin{align}\label{difference}
  \left(\Phi,[R_\eps(z) - R_0(z)]\Psi\right)_{\sii(\Real^d)}
  = \ & \left((H_0^*-\bar{z})\phi_0,\psi_\eps\right)_{\sii(\Real^d)}
  - \left(\phi_0,(H_\eps-z)\psi_\eps\right)_{\sii(\Real^d)}
  \nonumber \\
  = \ & h_0(\phi_0,\psi_\eps) - h_\eps(\phi_0,\psi_\eps)
  \nonumber \\
  =\ & \alpha_+ \left[
  (\phi_0,\psi_\eps)_{\sii(\Sigma_0)}-(\phi_0,\psi_\eps)_{\sii(\Sigma_{+\eps})}
  \right]
  \nonumber \\
  \ &+ \alpha_- \left[
  (\phi_0,\psi_\eps)_{\sii(\Sigma_0)}-(\phi_0,\psi_\eps)_{\sii(\Sigma_{-\eps})}
  \right]
  ,
\end{align}
where the second equality employs the fact that
the form domains of~$H_\eps$ and~$H_0$ coincide.
The boundary terms after the last equality
should be interpreted in the sense of traces~\eqref{trace}.

The unitary transform~\eqref{unitary} enables us to identify
$\sii(\Sigma_{\pm\eps})$ with $\sii(\Sigma_{0})$.
Writing $u_0 := \mathcal{U}\phi_0$ and $v_\eps := \mathcal{U}\psi_\eps$
and recalling~\eqref{elements},
we have
\begin{eqnarray*}
  \lefteqn{
  (\phi_0,\psi_\eps)_{\sii(\Sigma_0)}-(\phi_0,\psi_\eps)_{\sii(\Sigma_{+\eps})}
  }
  \\
  &&= \int_{\Sigma_0} (\bar{u}_0 v_\eps)(q,0) \, \D\Sigma_0
  - \int_{\Sigma_0} (\bar{u}_0 v_\eps)(q,\eps) \, f(q,\eps)\,\D\Sigma_0
  \\
  && = \underbrace{-\int_{\Sigma_0\times(0,\eps)}
  \partial_t(\bar{u}_0 v_\eps)(q,t) \, \D\Sigma_0 \wedge \D t}_{I_1}
  + \underbrace{\int_{\Sigma_0}
  (\bar{u}_0 v_\eps)(q,\eps) \, [1-f(q,\eps)]\,\D\Sigma_0}_{I_2}
  \,.
\end{eqnarray*}

Here the last integral can be estimated as follows
$$
\begin{aligned}
  |I_2|
  &\leq \|\phi_0\|_{\sii(\Sigma_{+\eps})} \, \|\psi_\eps\|_{\sii(\Sigma_{+\eps})}
  \, \sup_{q\in\Sigma_0} \frac{|1-f(q,\eps)|}{f(q,\eps)}
  \\
  &\leq C \, \|\phi_0\|_{H^1(\Real^d)} \, \|\psi_\eps\|_{H^1(\Real^d)}
  \, \sup_{q\in\Sigma_0} \frac{|1-f(q,\eps)|}{f(q,\eps)}
  ,
\end{aligned}
$$
where the second inequality is due to~\eqref{trace.norm}.
Taking into account the explicit formula for~$f$ in~\eqref{Jacobian}
and~\eqref{f.bound},
we see that there is a constant~$C$
(depending on the geometric number~$a$
and the supremum norms of the curvature functions~$K_\mu$)
such that
\begin{equation}\label{eq-supnorm}
  \sup_{q\in\Sigma_0} \frac{|1-f(q,\eps)|}{f(q,\eps)}
  \leq C \, \eps
  \,.
\end{equation}
By Lemma~\ref{Lem.res.bound1}, we have
\begin{equation}\label{Tsa}
  \|\psi_\eps\|_{H^1(\Real^d)} \leq C \, \|\Psi\|_{\sii(\Real^d)}
  \qquad \mbox{and} \qquad
  \|\phi_0\|_{H^1(\Real^d)} \leq C \, \|\Phi\|_{\sii(\Real^d)}
  \,.
\end{equation}
Since~$\phi_0$ is defined via the adjoint of the resolvent of~$H_0$,
it might be useful to mention for the latter inequality
that~$H_0$ satisfies the $\mathcal{T}$-self-adjointness relation
$H_0^* = \mathcal{T}H_0\mathcal{T}$,
where~$\mathcal{T}$ is the complex-conjugation operator.
Summing up,
\begin{equation}\label{I2}
  |I_2| \leq C \, \eps \,
  \|\Phi\|_{\sii(\Real^d)} \, \|\Psi\|_{\sii(\Real^d)}
  \,.
\end{equation}

We now turn to estimating~$I_1$.
First of all, we use the Schwarz inequality to get
$$
  |I_1| \leq
  \|u_0\|_{\sii(\Sigma_0\times(0,\eps))}
  \|\partial_t v_\eps\|_{\sii(\Sigma_0\times(0,\eps))}
  + \|\partial_t u_0\|_{\sii(\Sigma_0\times(0,\eps))}
  \|v_\eps\|_{\sii(\Sigma_0\times(0,\eps))}
  \,.
$$
Here the first term on the right hand side
can be estimated as follows
$$
\begin{aligned}
  \|u_0\|_{\sii(\Sigma_0\times(0,\eps))}^2
  &\leq \eps \sup_{t \in (0,\eps)} \int_{\Sigma_0}
  |u_0(q,t)|^2 \, \D \Sigma_0
  \leq C \, \eps \sup_{t \in (0,\eps)} \int_{\Sigma_t} |\phi_0|^2
  \,,
\end{aligned}
$$
where $C:=1/\inf_{\Sigma_0\times(0,a)} f$.
Using in addition~\eqref{trace.norm} and~\eqref{Tsa},
we eventually obtain
$$
  \|u_0\|_{\sii(\Sigma_0\times(0,\eps))}
  \leq C \, \sqrt{\eps} \, \|\Phi\|_{\sii(\Real^d)}
  \,.
$$
In the same manner, we get
$$
  \|v_\eps\|_{\sii(\Sigma_0\times(0,\eps))}
  \leq C \, \sqrt{\eps} \, \|\Psi\|_{\sii(\Real^d)}
  \,.
$$

The terms $\|\partial_t v_\eps\|_{\sii(\Sigma_0\times(0,\eps))}$
and $\|\partial_t u_0\|_{\sii(\Sigma_0\times(0,\eps))}$
require a bit more careful analysis.
As above, we write
$$
\begin{aligned}
  \|\partial_t u_0\|_{\sii(\Sigma_0\times(0,\eps))}^2
  &\leq \eps \sup_{t \in (0,\eps)} \int_{\Sigma_0}
  |\partial_t u_0(q,t)|^2 \, \D \Sigma_0
  \leq C \, \eps \sup_{t \in (0,\eps)} \int_{\Sigma_t} |\partial_n\phi_0|^2
  \,,
\end{aligned}
$$
where we have also used $\partial_t u_0 = \partial_n\phi_0 \circ \mathcal{L}$.
Now, however, we cannot use~\eqref{trace.norm}
because~$\phi_0$ is not in~$H^2(\Real^d)$.
Nevertheless, it belongs to $H^2(\Omega_0^{+})$,
where the set~$\Omega_0^{+}$ is defined~\eqref{Omega.sets}.
Hence,
$$
  \sup_{t \in (0,\eps)} \int_{\Sigma_t} |\partial_n\phi_0|^2
  \leq \sup_{t \in (0,a/2)} \int_{\Sigma_t} |\partial_n\phi_0|^2
  \leq C \, \|\phi_0\|_{H^2(\Omega_0^{+})}^2
  \,,
$$
where the last inequality is a trace embedding based on~\eqref{elementary}.
Applying Lemma~\ref{Lem.res.bound2}, we eventually get
the desired bound
$$
  \|\partial_t u_0\|_{\sii(\Sigma_0\times(0,\eps))}
  \leq C \, \sqrt{\eps} \, \|\Phi\|_{\sii(\Real^d)}
  \,.
$$

It remains to estimate $\|\partial_t v_\eps\|_{\sii(\Sigma_0\times(0,\eps))}$.
Still, as above, we could also write
$$
\begin{aligned}
  \|\partial_t v_\eps\|_{\sii(\Sigma_0\times(0,\eps))}^2
  &\leq \eps \sup_{t \in (0,\eps)} \int_{\Sigma_0}
  |\partial_t v_\eps(q,t)|^2 \, \D \Sigma_0
  \leq C \, \eps \sup_{t \in (0,\eps)} \int_{\Sigma_t} |\partial_n\psi_\eps|^2
  \,.
\end{aligned}
$$
Now, however, the situation is worse than for~$\phi_0$,
because~$\psi_\eps$ belongs only to $H^2(\Omega_\eps^{0+})$,
where
\begin{equation*}
  \Omega_{\eps}^{0+}
  := \left\{
  \mathcal{L}(q,t) : \, q \in \Sigma_0, \, 0 < t < \eps
  \right\}
  ,
\end{equation*}
is diminishing as $\eps \to 0$.
Consequently, \eqref{elementary}~would give
a bad $\eps$-dependent estimate on the norm of the trace operator
associated with the embedding
$
  H^2(\Omega_\eps^{0+}) \to  H^1(\Sigma_t)
$
with $t \in (0,\eps)$.
Instead, we integrate by parts
$$
\begin{aligned}
  \|\partial_t v_\eps\|_{\sii(\Sigma_0\times(0,\eps))}^2
  =\ &
  \int_{\Sigma_0\times(0,\eps)}
  (\partial_t t) \, |\partial_t v_\eps(q,t)|^2 \, \D \Sigma_0 \wedge \D t
  \\
  =\ & - \int_{\Sigma_0\times(0,\eps)}
  2\, t \, \Re \left[
  \partial_t \bar{v}_\eps(q,t) \, \partial_t^2 v_\eps(q,t)
  \right]
  \, \D \Sigma_0 \wedge \D t
  \\
  \ & + \eps
  \lim_{t \to \eps^-}
  \int_{\Sigma_0} |\partial_t v_\eps(q,t)|^2
  \, \D \Sigma_0
  \\
  \leq\ & C \, \eps \left(
  \|\psi_\eps\|_{H^2(\Omega_{\eps}^{0+})}^2
  + \|\tau_{+\eps}^-\partial_n\psi_\eps\|_{\sii(\Sigma_{+\eps})}^2
  \right)
  \,,
\end{aligned}
$$
where the inequality employs $t \leq \eps$
and the geometric estimates~\eqref{f.bound}
together with
$\partial_t^2 v_\eps = \partial_n^2\psi_\eps \circ \mathcal{L}$.
Recall that the trace operator~$\tau_{+\eps}^-$ is defined in~\eqref{traces}.
The trick is to replace $\tau_{+\eps}^-\partial_n\psi_\eps$
by $\tau_{+\eps}^+\partial_n\psi_\eps$
using the interface condition~\eqref{conditions}
and employ~\eqref{elementary} in
the other set that does not diminish as $\eps \to 0$:
$$
\begin{aligned}
  \|\tau_{+\eps}^-\partial_n\psi_\eps\|_{\sii(\Sigma_{+\eps})}
  &\leq \|\tau_{+\eps}^+\partial_n\psi_\eps\|_{\sii(\Sigma_{+\eps})}
  + |\alpha_+| \, \|\tau_{+\eps}\psi_\eps\|_{\sii(\Sigma_{+\eps})}
  \\
  &\leq C \left(
  \|\psi_\eps\|_{H^2(\Omega_\eps^+)}
  + |\alpha_+| \, \|\psi_\eps\|_{H^1(\Real^d)}
  \right)
  \,.
\end{aligned}
$$
Using Lemma~\ref{Lem.res.bound2}
and~\eqref{Tsa},
we eventually get the desired bound
$$
  \|\partial_t v_\eps\|_{\sii(\Sigma_0\times(0,\eps))}
  \leq C \, \sqrt{\eps} \, \|\Psi\|_{\sii(\Real^d)}
  \,.
$$

Summing up, we have proved
\begin{equation}\label{I1}
  |I_1| \leq C \, \eps \,
  \|\Phi\|_{\sii(\Real^d)} \, \|\Psi\|_{\sii(\Real^d)}
  \,.
\end{equation}
This bound together with~\eqref{I2} implies
$$
  \left|
  (\phi_0,\psi_\eps)_{\sii(\Sigma_0)}-(\phi_0,\psi_\eps)_{\sii(\Sigma_{+\eps})}
  \right|
  \leq C \, \eps \,
  \|\Phi\|_{\sii(\Real^d)} \, \|\Psi\|_{\sii(\Real^d)}
$$
and a similar estimates holds for the other difference
of boundary terms in~\eqref{difference}.
Consequently,
$$
  \left|
  \left(\Phi,[R_\eps(z) - R_0(z)]\Psi\right)_{\sii(\Real^d)}
  \right|
  \leq C \, \eps \,
  \|\Phi\|_{\sii(\Real^d)} \, \|\Psi\|_{\sii(\Real^d)}
  \,,
$$
which proves~\eqref{nrs} for $z \in \Com$ with $\Re z < z_0$.
The extension to other values of~$z$ is standard
(\cf~\cite[Sec.~IV.3.3]{Kato}).
\end{proof}
\begin{remark}
Taking into account Remark~\ref{Rem.duality},
\eqref{difference}~implies the operator identity
\begin{equation}\label{first}
  R_\eps(z) - R_0(z)
  = R_0(z)
  \left[
  (\alpha_+ + \alpha_-) \tau_0^*\tau_0
  - \alpha_+ \tau_{+\eps}^* \tau_{+\eps}
  - \alpha_- \tau_{-\eps}^* \tau_{-\eps}
  \right]
  R_\eps(z)
  \,,
\end{equation}
where
$$
  \tau_t^* : \sii(\Sigma_t) \to H^{-1}(\Real) :
  \{\psi \mapsto \psi \, \delta_{\Sigma_t}\}
  \,.
$$
It is a generalisation of
the \emph{first resolvent identity} known for regular potentials.
\end{remark}
%

\section{Convergence of eigenvalues and eigenfunctions}\label{Sec.efs}
%
In this section, we deduce from Theorem~\ref{Thm.nrs}
a convergence of eigenvalues and eigenfunctions of~$H_\eps$
to eigenvalues and eigenfunctions of~$H_0$ as $\eps \to 0$.
In fact, it is immediately seen that the eigenfunctions
converge in the topology of $\sii(\Real^d)$.
By using the maximum principle in a refined way,
we show the non-trivial property that the convergence
actually holds uniformly in a neighbourhood of~$\Sigma_0$.
This result will be needed in Section~\ref{Sec.evs}
to prove Theorem~\ref{Thm.evs}.

Let $\lambda_0$ stand for a simple eigenvalue of~$H_0$
with the corresponding eigenfunction~$\psi_0$ which is assumed to be normalised
according to the usual requirement for non-self-adjoint
spectral problems
$$
  \big(\overline{\psi_0}, \psi_0\big)_{L^2(\R^d)}
  =\int_{\R^d} \psi_0^2 = 1
  \,.
$$
By a simple eigenvalue we always mean that
of algebraic multiplicity one.
Note that~$\overline{\psi_0}$ represents an eigenfunction
of the adjoint operator~$H_0^*$ corresponding
to the eigenvalue~$\overline{\lambda_0}$.
Define
\begin{equation}\label{circle}
  \mathcal{C}_r := \{ z\in \C \,:\, |z-\lambda_0 | = r  \}
  \,,
\end{equation}
where the radius~$r $ is chosen is such a way
that the circle~$\mathcal{C}_r$ surrounds only one point of
$\sigma_{\mathrm{disc}}(H_0)$,
the discrete spectrum of~$H_0$.
The resolvent convergence proved in the previous
section allows us to claim that there exists $\varepsilon_0 >0$ such that
for any non-negative $\varepsilon < \varepsilon_0 $
the circle $\mathcal{C}_r$ surrounds only one
point $\lambda_\varepsilon$  of
$\sigma_{\mathrm{disc}} (H_\varepsilon )$. Let $P_\varepsilon$ stand for
the eigenprojector
\begin{equation}\label{projector}
P_\varepsilon :=\frac{i}{2\pi }\oint_{\mathcal{C}_r }
R_\varepsilon (z) \, \mathrm{d}z\,,
\end{equation}
where  the integration path   traces out the circle
around in a counterclockwise manner.
Let $\psi_\varepsilon$ stand for the  eigenfunction of $H_\varepsilon$ corresponding
to $\lambda_\varepsilon$ and impose the same normalisation condition
$\big(\overline{\psi_\varepsilon}, \psi_\varepsilon \big)_{L^2(\R^d)}= 1$.
Then the corresponding eigenprojector takes the form
$$
  P_\varepsilon = \big(\overline{\psi_\varepsilon} , \cdot \big)_{L^2 (\R^d)}
  \, \psi_\varepsilon \,.
$$
The following statement is  a simple consequence of
the norm-resolvent convergence
(Theorem~\ref{Thm.nrs})
proved in the previous section.
\begin{corollary}\label{cor-2}
The asymptotics
\begin{equation}\label{eq-P}
  \|P_\varepsilon - P_0\|_{L^2 (\R^d )\to L^2 (\R^d )} = O(\varepsilon )\,
\end{equation}
holds. Consequently, we have
\begin{equation}\label{eq-embedingesta}
  |\lambda_\eps - \lambda_0| =O(\varepsilon)
  \qquad\mbox{and}\qquad
  \| \psi_\varepsilon -\psi_0 \|_{L^2(\R^d )}
  =O(\varepsilon)\,. 
\end{equation}
\end{corollary}

The rest of this section is devoted to showing
that the convergence of eigenfunctions holds in a better topology,
at least in a neighbourhood of~$\Sigma_0$.
First of all, we establish a regularity of eigenfunctions.
\begin{proposition}\label{Lem.regularity}
Given $\eps \geq 0$, let~$\psi_\eps$ denote
an eigenfunction of~$H_\eps$. Then
\begin{equation}\label{eq-ellregef}
  \psi_\varepsilon \in
  H^m\big(\R^d \setminus (\Sigma_{+\varepsilon} \cup \Sigma_{-\varepsilon})\big)
  \qquad \mbox{for all} \qquad m\in \N
  \,.
\end{equation}
\end{proposition}
\begin{proof}
We have $H_\eps\psi_\eps = \lambda_\eps \psi_\eps$,
where $\lambda_\eps \in \Com$ is the eigenvalue
and $\psi_\eps \in \Dom(H_\eps)$.
For $m=2$ the claim of the lemma follows from the characterisation
of the operator domain~\eqref{Hamiltonian.domain}.
Starting from the definition of the operator~$H_\eps$
through its quadratic form~\eqref{form} defined
on the Sobolev space~$H^1(\Real^d)$,
the $H^2$-regularity outside
$\Sigma_{+\varepsilon} \cup \Sigma_{-\varepsilon}$
is actually established by our Lemma~\ref{Lem.res.bound2}.
For the present eigenvalue problem,
we can write
\begin{equation}\label{iterate}
  (H_\eps-z)\psi_\eps
  = (\lambda_\eps-z) \psi_\eps
  =: \Psi_\eps
  \,,
\end{equation}
where $z$~is any number from the resolvent set of~$H_\eps$.
Recalling that~$H_\eps$ acts as the Laplacian
outside $\Sigma_{+\varepsilon} \cup \Sigma_{-\varepsilon}$,
from elliptic regularity theory (see, \eg, \cite[Thm.~6.3.2]{Evans}),
we immediately get
$
  \psi_\eps \in H^m (\Real^d \setminus \overline{\Omega_{a/4}^0})
$
for all $m \in \N$.
It remains to show the $H^m$-regularity close
to the parallel hypersurfaces
$\Sigma_{+\varepsilon} \cup \Sigma_{-\varepsilon}$.

Let us comment on the proof for $\eps=0$.
The case of positive~$\eps$ is proved analogously.
We refer to \cite[Sec.~6.3]{Evans} for more details
on this type of elliptic-regularity-type arguments.
Setting $v := \mathcal{U} \psi_0$ and
$V := \mathcal{U} \Psi_0 = (\lambda_\eps-z) v$,
where~$\mathcal{U}$ is the unitary transform~\eqref{unitary}
implementing the natural curvilinear coordinates
in a vicinity of~$\Sigma_0$,
\eqref{iterate}~yields a weak formulation of the problem
\begin{equation}\label{iterate.curved}
\left\{
\begin{aligned}
  (-\Delta_G-z)v
  &= V
  &&\mbox{in} \quad
  \Sigma_{0} \times (-a,a)
  \,,
  \\
  v(q,0^+) - v(q,0^-)
  &= 0
  &&\mbox{on} \quad
  \Sigma_{0}
  \,,
  \\
  \partial_t v(q,0^+) - \partial_t v(q,0^-)
  &= (\alpha_+ +\alpha _- ) v(q,0)
  &&\mbox{on} \quad
  \Sigma_{0}
  \,,
\end{aligned}
\right.
\end{equation}
where the Laplace-Beltrami operator~$-\Delta_G$ acts as in~\eqref{LB}.
Once we know that the right hand side~$V$ belongs to
$H^2\big(\Sigma_0 \times [(-a,0)\cup(0,a)]\big)$,
we can differentiate~\eqref{iterate.curved}
(in the sense of weak derivatives)
and obtain that the derivative $\partial_\rho v$
with $\rho \in \{1,\dots,d-1\}$
again satisfies the same problem~\eqref{iterate.curved},
including the same interface conditions,
but with a changed right hand side $V' \in \sii(\Sigma_0 \times (-a/2,a/2))$.
By applying Lemma~\ref{Lem.res.bound2}, we deduce
$\partial_\rho v \in H^2\big(\Sigma_0 \times [(-a/2,0)\cup(0,a/2)]\big)$.
The fact that also respective restrictions of~$\partial_t^3 v$
belong to $\sii(\Sigma_0 \times (0,a/2))$
and $\sii(\Sigma_0 \times (-a/2,0))$
can be then shown from the differential equation
that~$\partial_\rho v$ satisfies almost everywhere,
by writing as in~\eqref{block-diagonal}.
Hence, we have established
$
  v \in H^3\big(\Sigma_0 \times [(-a/2,0)\cup(0,a/2)]\big)
$.
In particular, $V' \in H^2(\Sigma_0 \times (-a/2,a/2))$.
Repeating this argument, we eventually obtain
$
  v \in H^m\big(\Sigma_0 \times [(-a/2,0)\cup(0,a/2)]\big)
$
for all $m \in \N$.
\end{proof}

The proposition has the usual corollary that the eigenfunctions
are smooth outside the interface hypersurfaces.
\begin{corollary}\label{Corol.regularity}
Let~$\psi_\eps$ denote an eigenfunction of~$H_\eps$.
Then~$\psi_\eps$ is continuous in~$\Real^d$ and
\begin{equation}\label{regularity}
  \psi_\eps \in
\begin{cases}
  C^{\infty}\big(\overline{\Omega}\big)
  \cap C^{\infty}\big(\overline{\Real^d\setminus\Omega}\big)
  &\mbox{if} \quad \eps = 0 \,,
  \\
  C^{\infty}\big(\overline{\Omega_\eps^0}\big)
  \cap C^{\infty}\big(\overline{\Real^d\setminus\Omega_\eps^0}\big)
  &\mbox{if} \quad \eps > 0 \,.
\end{cases}
\end{equation}
\end{corollary}
\begin{proof}
By Proposition~\ref{Lem.regularity},
we have
$
  \psi_0
  \in H^m(\R^d \setminus \Sigma_0)
$
for every positive integer~$m$.
Hence, by the Sobolev embedding theorem
(see, \eg, \cite[Thm.~5.4]{Adams}),
$
  \psi_0 \in C^{k}(\overline{\Omega})
  \cap C^{k}(\overline{\Real^d \setminus \Omega})
$
for each positive integer~$k$.
This proves~\eqref{regularity} for $\eps =0$.
The continuity follows from the fact that~$\psi_0$
as an element of the form domain~$\Dom(h_0)$
belongs to $H^1(\Real^d)$.
The claims for positive~$\eps$ are proved analogously.
\end{proof}
As a consequence of this corollary,
the eigenvalue problem $H_\eps \psi_\eps = \lambda_\eps \psi_\eps$
can be considered in a classical sense.
Setting
\begin{equation}\label{efs.diff}
  \phi_\eps := \psi_\eps - \psi_0
\end{equation}
and combining the eigenvalue equations for $\eps>0$ and $\eps=0$,
we see that~$\phi_\eps$ with positive~$\eps$
is a continuous and piecewise smooth solution of
the classical boundary value problem
\begin{equation}\label{classical}
\left\{
\begin{aligned}
  - \Delta \phi_\eps - \lambda_\eps \phi_\eps
  &= (\lambda_\eps - \lambda_0) \psi_0
  && \mbox{in}\quad
  \Real^d \setminus (\Sigma_{+\eps} \cup \Sigma_{-\eps}\cup \Sigma_0)
  \,,
  \\
  \tau_{\pm\eps}^+ \partial_n \phi_\eps
  - \tau_{\pm\eps}^- \partial_n \phi_\eps
  - \alpha_\pm \tau_{\pm\eps} \phi_\eps
  &= \alpha_\pm \tau_{\pm\eps} \psi_0
  && \mbox{on} \quad
  \Sigma_{\pm \eps}
  \,,
  \\
  \tau_{+0}^+ \partial_n \phi_\eps
  - \tau_{-0}^- \partial_n \phi_\eps
  &= -(\alpha_+ + \alpha_-) \tau_{0} \psi_0
  && \mbox{on} \quad
  \Sigma_{0}
  \,.
\end{aligned}
\right.
\end{equation}

To establish the uniform convergence of eigenfunctions,
we use the maximum principle following the ideas of~\cite{FK4}.
The first ingredient is a version of the mean value theorem
in the present setting.
\begin{lemma}\label{Lem.MVT}
For every $x \in \Real^d$ and $r > 0$, we have the identity
\begin{eqnarray}\label{MVT}
  \lefteqn{\phi_\eps(x) =
  \frac{1}{|\partial B_r|} \int_{\partial B_r} \phi_\eps}
  \\
  &&
  + \int_0^r \frac{\D\rho}{|\partial B_\rho|}
  \Bigg[
  \lambda_\eps \int_{B_\rho} \phi_\eps
  - \alpha_+ \int_{\Sigma_{+\eps}^\rho} \phi_\eps
  - \alpha_- \int_{\Sigma_{-\eps}^\rho} \phi_\eps
  \nonumber \\
  && +
  (\lambda_\eps-\lambda_0) \int_{B_\rho} \psi_0
  - \alpha_+ \int_{\Sigma_{+\eps}^\rho} \psi_0
  - \alpha_- \int_{\Sigma_{-\eps}^\rho} \psi_0
  + (\alpha_+ + \alpha_-) \int_{\Sigma_{0}^\rho} \psi_0
  \Bigg]
  \,,
  \nonumber
\end{eqnarray}
where~$\phi_\eps$ denotes the difference of eigenfunctions~\eqref{efs.diff},
$B_r \equiv B_r(x)$ is the open ball
of radius~$r$ centred at~$x$,
$|\partial B_r|$~stands for the $(d-1)$-dimensional
Hausdorff measure of its boundary
and $\Sigma_{\pm\eps}^r := \Sigma_{\pm\eps} \cap B_r$.
\end{lemma}
\begin{proof}
The formula follows by integrating the differential equation of~\eqref{classical}
in the ball~$B_\rho$ of radius $\rho \in (0,r]$,
using the interface conditions of~\eqref{classical}
after an application of the divergence theorem
and handling the boundary term $\int_{\partial B_\rho} \partial\phi_\eps/\partial\nu$,
with~$\nu$ denoting the outward unit normal to~$\partial B_\rho$	,
as in the classical mean value theorem, see \cite[Thm.~2.1]{Gilbarg-Trudinger}.
\end{proof}

To handle the first term on the right hand side of~\eqref{MVT},
we use the following elementary result (\cite[Lem.~3.14]{FK4}).
\begin{lemma}\label{Lem.Fubini}
Let $\phi \in \sii(\Real^d)$ and $\delta > 0$.
For every $x \in \Real^d$,
there exists $r =r(x,\phi,\delta) \in (0,\delta]$ such that
$$
  \frac{1}{|\partial B_r|} \int_{\partial B_r} |\phi|
  \leq \frac{1}{|B_\delta|^{1/2}} \,
  \|\phi\|_{\sii(B_\delta)}
  \,.
$$
Here~$|B_r|$ denotes the $d$-dimensional Lebesgue measure
of the ball~$B_r$.
\end{lemma}
\begin{proof}
Assume by contradiction that there exists a point $x \in \Real^d$
such that for all $r \in (0,\delta]$ the reverse inequality holds.
Then one easily arrives at a contradiction by using in addition
the coarea formula and the Schwarz inequality.
\end{proof}

Now we are in a position to establish the uniform convergence of eigenfunctions.
While Proposition~\ref{Lem.regularity}
and its Corollary~\ref{Corol.regularity}
deal with \emph{any} eigenfunctions of~$H_\eps$,
from now on we assume again that~$\psi_\eps$ and~$\psi_0$
are eigenfunctions of~$H_\eps$ and~$H_0$, respectively,
corresponding to simple eigenvalues~$\lambda_\eps$ and~$\lambda_0$
as described in the beginning of this section.
\begin{theorem}\label{Thm.uniform}
We have
\begin{equation}\label{uniform}
  \|\psi_\eps - \psi_0\|_{L^\infty(\Sigma_{\pm\eps})} = O(\eps)
  \,.
\end{equation}
\end{theorem}
\begin{proof}
Assume that $0 < \eps \leq \delta/2$,
where $\delta<a$ is a positive number independent of~$\eps$
that will be additionally restricted later on.
From~\eqref{classical} and the methods of the theory of interior regularity
of solutions of elliptic problems
(see, \eg, \cite[Sec.~6.3.1]{Evans}),
we deduce the bound
$$
  \|\phi_\eps\|_{H^{m+2}(\Real^d\setminus\overline{\Omega_\delta^0})}
  \leq C \left(
  \|\phi_\eps\|_{\sii(\Real^d)}
  + |\lambda_\eps-\lambda_0| \|\psi_0\|_{H^m(\Real^d\setminus\Sigma_0)}
  \right)\,
$$
for every $m \in \N$.
Here the constant~$C$ depends on~$d$, $\delta$ and~$m$,
but it is independent of~$\eps$
(the dependence of the coefficient~$\lambda_\eps$ on~$\eps$
on the left hand side of the differential equation in~\eqref{classical}
is unimportant due to Corollary~\ref{cor-2}).
Within this proof, the symbol~$C$ denotes a generic constant
whose value may change from line to line,
but it is always independent of~$\eps$.
By the convergence results of Corollary~\ref{cor-2},
the regularity of Proposition~\ref{Lem.regularity}
and the Sobolev embedding theorem, we obtain
\begin{equation}\label{interior}
  \|\phi_\eps\|_{C^k(\overline{\Real^d\setminus\Omega_\delta^0}))}
  \leq C \, \eps
\end{equation}
for every $k \in \N$.
In particular, this proves the uniform convergence
of eigenfunctions in $\Real^d\setminus\Omega_\delta^0$.
To prove the uniform convergence in a neighbourhood of~$\Sigma_0$
containing the colliding hypersurfaces~$\Sigma_{+\eps}$ and~$\Sigma_{-\eps}$,
we give slightly different proofs in high and low dimensions.

\medskip
\noindent
\fbox{$d \geq 3$}
First of all, we employ Lemma~\ref{Lem.MVT}
with $x \in \Omega_\delta^0$ and $r \leq \delta$.
We estimate the terms on the right hand side of~\eqref{MVT} as follows.
For every continuous function $\phi \in L^\infty(\Real^d)$,
we have
$$
\begin{aligned}
  \left|
  \int_0^r \frac{\D\rho}{|\partial B_\rho|} \int_{B_\rho} \phi
  \right|
  &\leq
  \|\phi\|_{L^\infty(\Omega_{2\delta}^0)}
  \int_0^r \frac{|B_\rho|}{|\partial B_\rho|} \, \D\rho
  = \|\phi\|_{L^\infty(\Omega_{2\delta}^0)} \, \frac{r^2}{2 d}
  \,,
  \\
  \left|
  \int_0^r \frac{\D\rho}{|\partial B_\rho|}
  \int_{\Sigma_{\pm\eps}^\rho} \phi
  \right|
  &\leq
  \|\phi\|_{L^\infty(\Sigma_{\pm\eps})}
  \int_0^r \frac{|\Sigma_{\pm\eps}^\rho|}{|\partial B_\rho|} \, \D\rho
  \leq C \, \|\phi\|_{L^\infty(\Sigma_{\pm\eps})} \, r
  \,.
\end{aligned}
$$
Here the last estimate employs the geometric bound
$
  |\Sigma_{\pm\eps}^\rho| \leq C \rho^{d-1}
$.
Consequently, using Corollary~\ref{cor-2},
\begin{equation}\label{MVT.bounds1}
\begin{aligned}
  \left| \lambda_\eps
  \int_0^r \frac{\D\rho}{|\partial B_\rho|} \int_{B_\rho} \phi_\eps
  \right|
  &\leq C \, \|\phi_\eps\|_{L^\infty(\Omega_{2\delta}^0)} \, \delta^2
  \,,
  \\
  \left| (\lambda_\eps-\lambda_0)
  \int_0^r \frac{\D\rho}{|\partial B_\rho|} \int_{B_\rho} \psi_0
  \right|
  &\leq C \, \eps
  \,,
  \\
  \left| \alpha_\pm
  \int_0^r \frac{\D\rho}{|\partial B_\rho|}
  \int_{\Sigma_{\pm\eps}^\rho} \phi_\eps
  \right|
  &\leq C \, \|\phi_\eps\|_{L^\infty(\Sigma_{\pm\eps})} \, \delta
  \,.
\end{aligned}
\end{equation}
To handle the last terms on the right hand side of~\eqref{MVT},
we recall the unitary transform~\eqref{unitary}.
Setting $v_0 := \mathcal{U} \psi_0$, we have
\begin{eqnarray*}
  \lefteqn{
  \int_{\Sigma_{\eps}^\rho} \psi_0 - \int_{\Sigma_0^\rho} \psi_0
  = \int_{p_\eps^{-1}(\Sigma_{\eps}^\rho)}
  v_0(q,\eps) \, f(q,\eps) \, \D q
  - \int_{\Sigma_0^\rho} v_0(q,0) \, \D q
  }
  \\
  && = \int_{p_\eps^{-1}(\Sigma_{\eps}^\rho) \cap \Sigma_0^\rho}
  \int_0^\eps \partial_t(v_0 f)(q,t)  \, \D t \, \D q
  \\
  && \qquad
  + \int_{p_\eps^{-1}(\Sigma_{\eps}^\rho) \setminus \Sigma_0^\rho}
  v_0(q,\eps) \, f(q,\eps) \, \D q
  - \int_{\Sigma_0^\rho \setminus p_\eps^{-1}(\Sigma_{\eps}^\rho)} v_0(q,0) \, \D q
  \,,
\end{eqnarray*}
where $p_\eps(q) := \mathcal{L}(q,\eps)$.
Consequently,
\begin{align*}
  \left|
  \int_{\Sigma_{\eps}^\rho} \psi_0 - \int_{\Sigma_0^\rho} \psi_0
  \right|
  \leq \ &  |\Sigma_0| \, \eps \,
  \|v_0\|_{C^1(\overline{\Sigma_0\times(0,\delta)})}
  \, \|f\|_{C^1(\overline{\Sigma_0\times(0,\delta)})}
  \\
  & +
  \big|p_\eps^{-1}(\Sigma_{\eps}^\rho) \,\triangle\, \Sigma_0^\rho\big| \,
  \|v_0\|_{C^0(\overline{\Sigma_0\times(0,\delta)})}
  \, \|f\|_{C^0(\overline{\Sigma_0\times(0,\delta)})}
  \,.
\end{align*}
It is a matter of purely geometric considerations to check
that the estimate
\begin{equation}\label{fundamental}
  \big|p_\eps^{-1}(\Sigma_{\eps}^\rho) \,\triangle\, \Sigma_0^\rho\big|
  \leq C \, \eps^{(d-1)/2}
\end{equation}
holds true.
Hence, in view of~\eqref{Jacobian} and Corollary~\ref{Corol.regularity},
we get the estimate
\begin{equation}\label{MVT.bounds2}
  \left|
  \int_{\Sigma_{\eps}^\rho} \psi_0 - \int_{\Sigma_0^\rho} \psi_0
  \right|
  \leq C \, \eps
  \,.
\end{equation}
The same bound holds for $\Sigma_{-\eps}^\rho$
instead of~$\Sigma_{\eps}^\rho$.
Summing up, using the estimates~\eqref{MVT.bounds1} and~\eqref{MVT.bounds2}
in~\eqref{MVT} and assuming that $\delta \leq 1$,
we arrive at
\begin{equation}\label{MVT.consequence}
  |\phi_\eps(x)| \leq
  \frac{1}{|\partial B_r|} \int_{\partial B_r} |\phi_\eps|
  + C \, \eps
  + C \, \|\phi_\eps\|_{L^\infty(\Omega_{2\delta}^0)} \, \delta
  \,.
\end{equation}

Let $x_\eps \in \Omega_\delta^0$ be a point
in which~$|\phi_\eps|$ achieves its maximum in~$\overline{\Omega_\delta^0}$,
\ie\ $\sup_{x\in\Omega_\delta^0}|\phi_\eps(x)|=|\phi_\eps(x_\eps)|$.	
We write
$$
  \|\phi_\eps\|_{L^\infty(\Omega_{2\delta}^0)}
  \leq \|\phi_\eps\|_{L^\infty(\Omega_{\delta}^0)}
  + \|\phi_\eps\|_{L^\infty(\Omega_{2\delta}^0\setminus\Omega_{\delta}^0)}
  \leq |\phi_\eps(x_\eps)| + C \, \eps
  \,,
$$
where the second inequality follows from~\eqref{interior}.
Using this estimate in~\eqref{MVT.consequence}, we obtain
\begin{equation}\label{MVT.consequence2}
  (1 - C\delta) \, |\phi_\eps(x_\eps)|
  \leq
  \frac{1}{|\partial B_r|} \int_{\partial B_r} |\phi_\eps|
  + C \, \eps
  \,.
\end{equation}
Consequently, choosing~$\delta$ sufficiently small
in comparison to the constant~$C$
on the left hand side (coming from~\eqref{interior}),
we arrive at
\begin{equation}\label{MVT.consequence3}
  \|\phi_\eps\|_{L^\infty(\Omega_{\delta}^0)}
  = |\phi_\eps(x_\eps)| \leq
  \frac{C}{|\partial B_r|} \int_{\partial B_r} |\phi_\eps|
  + C \, \eps
  \,.
\end{equation}

Finally, applying Lemma~\ref{Lem.Fubini}
to the right hand side of~\eqref{MVT.consequence3},
we get
$$
  \|\phi_\eps\|_{L^\infty(\Omega_{\delta}^0)}
  \leq \frac{C}{|B_\delta|^{1/2}} \,
  \|\phi_\eps\|_{\sii(B_\delta)}
  + C \, \eps
  \,.
$$
By Corollary~\ref{cor-2} and~\eqref{interior},
we obtain the uniform convergence
\begin{equation}\label{uniform.high}
  \|\phi_\eps\|_{L^\infty(\Real^d)}
  \leq C \, \eps
  \,,
\end{equation}
which in particular implies~\eqref{uniform}.

\medskip
\noindent
\fbox{$d = 2$}
The above proof fails in low dimensions,
because~\eqref{fundamental} does not give the desired
decay rate of order~$\eps$.
In dimension $d=2$, however, just a slight modification
is needed to repair it
by noticing that the better estimate
\begin{equation}\label{fundamental.better}
  \big|p_\eps^{-1}(\Sigma_{\eps}^\rho) \,\triangle\, \Sigma_0^\rho\big|
  \leq C \, \eps^{d-1}
\end{equation}
holds (in all dimensions) provided that the centre~$x$ of the ball~$B_r$
is chosen within a distance of order~$\eps$ from~$\Sigma_0$.
More specifically, we choose $x \in \Omega_{2\eps}^0$.
Then~\eqref{MVT.bounds2} does hold even if $d=2$.
At the same time, the first term in~\eqref{MVT.bounds1}
must be handled differently;
we use the Schwarz inequality to get
\begin{equation}\label{MVT.bounds1.low}
  \left| \lambda_\eps
  \int_0^r \frac{\D\rho}{|\partial B_\rho|} \int_{B_\rho} \phi_\eps
  \right|
  	\leq C \, \|\phi_\eps\|_{L^\infty(\Omega_{2\eps}^0)} \,
  \int_0^r \frac{|B_\rho|^{1/2}}{|\partial B_\rho|} \, \D\rho
  \,,
\end{equation}
where the integral on the right hand side equals $r/(2\sqrt{\pi})$.
Consequently, estimate~\eqref{MVT.consequence}
can be replaced by
\begin{equation}\label{MVT.consequence.low}
  |\phi_\eps(x)| \leq
  \frac{1}{|\partial B_r|} \int_{\partial B_r} |\phi_\eps|
  + C \, \eps
  + C \, \|\phi_\eps\|_{L^\infty(\Omega_{2\eps}^0)} \, \delta
  \,.
\end{equation}
Choosing now $x_\eps \in \Omega_{2\eps}^0$ to be a point
in which~$|\phi_\eps|$ achieves its maximum in~$\overline{\Omega_{2\eps}^0}$,
we again get the estimate~\eqref{MVT.consequence2}
and applying Lemma~\ref{Lem.Fubini} together with Corollary~\ref{cor-2},
we obtain
\begin{equation}\label{uniform.low}
  \|\phi_\eps\|_{L^\infty(\Omega_{2\eps}^0)}
  \leq C \, \eps
  \,.
\end{equation}
In particular, it implies~\eqref{uniform}.

\medskip
\noindent
\fbox{$d = 1$}
We do not see a way how to make the present proof work in dimension $d=1$,
where even~\eqref{fundamental.better} gives just a uniform bound,
so we get no decay in~$\eps$ for the left hand side of~\eqref{MVT.bounds2}.
In the one-dimensional situation, however,
the eigenvalue problem is explicitly solvable (see Appendix)
and it can be checked by hand
that the uniform convergence~\eqref{uniform.high} holds.
\end{proof}
\begin{remark}
We point out that the previous proof gives
the  uniform convergence of eigenfunctions~\eqref{uniform.high}
in the whole~$\Real^d$ with $d \geq 3$.
It holds also if $d=1$ by an explicit verification.
If $d=2$, we only get~\eqref{uniform.low} and~\eqref{interior}
(these results holds in all dimensions, of course)
and the global bound~\eqref{uniform.high}
with~$\eps$ being replaced by~$\sqrt{\eps}$ on the right hand side.
\end{remark}

As a consequence of Theorem~\ref{Thm.uniform},
we get the following lemma that will be needed
in the next section.
\begin{lemma}\label{Lem.derivative}
We have
$$
 \int_{\Sigma_{\pm \eps}} \psi_0 \, \partial_n^\pm (\psi_\eps-\psi_0)
  = O(\eps)
  \,.
$$
\end{lemma}
\begin{proof}
Let $\xi \in C_0^\infty(\Omega_a^0)$ be a real-valued function
such that $\xi=1$ on
$
  \Omega_\eps^+ \equiv
  \{\mathcal{L}(q,t) : \, q \in \Sigma_0, \, \eps < t < a/2\}
$,
\cf~\eqref{Omega.sets}.
Multiplying~\eqref{classical} by~$\xi\psi_0$
and integrating by parts over the larger set
$
  \tilde{\Omega}_\eps^+ :=
  \{\mathcal{L}(q,t) : \, q \in \Sigma_0, \, \eps < t < a\}
$,
we arrive at the identity
\begin{multline*}
  -\int_{\tilde{\Omega}_\eps^+} \Delta(\xi\psi_0) \, \phi_\eps
  - \int_{\Sigma_\eps} \partial_n^+\psi_0 \, \phi_\eps
  + \int_{\Sigma_\eps} \psi_0 \, \partial_n^+\phi_\eps
  - \lambda_\eps \int_{\tilde{\Omega}_\eps^+} \xi \psi_0 \, \phi_\eps
  \\
  = (\lambda_\eps-\lambda_0) \int_{\tilde{\Omega}_\eps^+} \xi \psi_0^2
  \,.
\end{multline*}
From Corollary~\ref{cor-2} and Theorem~\ref{Thm.uniform}
together with Corollary~\ref{Corol.regularity},
we thus deduce
$$
  \int_{\Sigma_{+\eps}} \psi_0 \, \partial_n^+ \phi_\eps
  = O(\eps)
  \,.
$$
This proves the claim for~$\Sigma_{+\eps}$.
The other asymptotics is proved analogously.
\end{proof}
%

\section{Eigenvalue asymptotics}\label{Sec.evs}
%
This section is devoted to a proof of Theorem~\ref{Thm.evs}
and its extension to degenerate eigenvalues.

\subsection{Simple eigenvalues}
The analysis of the eigenvalue asymptotics will be based on the formula
\begin{equation}\label{eq-formev}
\lambda_\varepsilon  =
  \frac{h_\varepsilon \big(\overline{P_\varepsilon\psi_0}, P_\varepsilon\psi_0 \big)}
  {\big(\overline{P_\varepsilon\psi_0}, P_\varepsilon \psi_0 \big)_{L^2 (\R^d)}}
  \,,
\end{equation}
where
\begin{equation}\label{eq-aux1-}
  h_\varepsilon \big(\overline{P_\varepsilon \psi_0}, P_\varepsilon\psi_0 \big)
  = h_0 \big( \overline{\psi_0} , \psi_0 \big)
  + (h_\varepsilon - h_0 )\big( \overline{\psi_0} , \psi_0 \big)
  - h_\varepsilon
  \big(\overline{P_\varepsilon^\perp  \psi_0}, P_\varepsilon^\perp  \psi_0 \big)
\end{equation}
and
$$
 P_\varepsilon^\perp :=I -  P_\varepsilon\,.
$$
Note that the analogous decomposition was also a starting point for the eigenvalues analysis
derived in~\cite{EY} and~\cite{Exner-Kondej_2008}.
However,  our further strategy
is based on  essentially different arguments.
In particular, it requires certain modifications
to the non-self-adjoint class of operators considered in this paper.

The first term on the right hand side of~(\ref{eq-aux1-}) yields
$ h_0 (\overline{\psi_0} , \psi_0) = \lambda_0$.
The following statement will allow
to estimate the second term.
\begin{proposition}\label{prop-1}
Suppose
$
  \psi \in H^{1}(\R^d)
  \cap C^{\infty}(\overline{\Omega_0^+})
  \cap C^{\infty}(\overline{\Omega_0^-})
$.
Then we have
\begin{multline}\label{eq-formconv}
  h_\varepsilon (\overline{\psi}, \psi)  - h_0 (\overline{\psi}, \psi)
  \\
  = \eps \left( \alpha_+ \int_{\Sigma_0} \partial^+_n \psi ^2
  +\alpha_- \int_{\Sigma_0} \partial^-_n \psi^2
  -(\alpha_+ -\alpha_-)(d-1) \int_{\Sigma_0 }  K_1 \,\psi ^2 \right)
  + O(\varepsilon^2 )
  \,,
\end{multline}
where the error term depends on $\psi$.
\end{proposition}
\begin{proof}
Similarly as above,
we define $v := \mathcal{U} \psi$,
which reflects the continuity properties of~$\psi$.
A straightforward calculation yields
\begin{align} \nonumber
  h_\varepsilon (\overline{\psi}, \psi)  - h_0 (\overline{\psi}, \psi)
  = \ &
  \alpha_+
  \int_{\Sigma_0}   v (q, \varepsilon )^2 \, f(q, \varepsilon )
  \, \mathrm{d}\Sigma_0
  \\ \nonumber
  & +\alpha_-
  \int_{\Sigma_0}   v (q, -\varepsilon )^2 \, f(q, -\varepsilon )
  \, \mathrm{d}\Sigma_0
  \\
  \label{eq-formconv1}
  &  -(\alpha_++\alpha_- )
  \int_{\Sigma_0}   v (q, 0 )^2 \, \mathrm{d}\Sigma_0
  \,.
\end{align}
Employing the continuity properties of~$v$,
we can expand
$$
  v (q, \pm \varepsilon )
  = v (q, 0 )\pm \varepsilon \, \partial_t v (q, 0^\pm  )
  + \breve{v}_\varepsilon
  \,,
$$
where
$
  \|\breve{v}_\varepsilon\|_{L^2 (\Sigma_0)}=O(\varepsilon^2)
$.
Applying these asymptotics
to~(\ref{eq-formconv1}) and combining it with~(\ref{Jacobian}),
we get the sought statement.
\end{proof}
The third term of~(\ref{eq-aux1-}) is estimated
by means of the following lemma.
\begin{lemma} \label{le-as3}
The  asymptotics
\begin{equation}\label{eq-estim2.bis}
h_\varepsilon \big(\overline{P_\varepsilon^\perp  \psi_0},
P_\varepsilon^\perp  \psi_0\big) =   \varepsilon
\left( (\alpha_+^2+\alpha_-^2 )\int_{\Sigma_0  }\psi_0 ^2 
 \right) + O(\varepsilon^2)\,
\end{equation}
holds, where  the error term depends on~$\psi_0$.
\end{lemma}
\begin{proof}
Let us denote
$$
\eta_\varepsilon (z):= \frac{i}{2\pi }\big(R_\varepsilon (z) - R_0
(z)\big)\psi_0  \,.
$$
Then
$$P_\varepsilon^\perp  \psi_0  = \int_{C_r }\eta_\varepsilon
 (z)  \, \mathrm{d}z\,.
$$
A straightforward calculation yields
\begin{align} \nonumber
h_\varepsilon
\big(\overline{P_\varepsilon^\perp  \psi_0}, P_\varepsilon^\perp  \psi_0\big)
= \ &
  \int_{\mathcal{C}_r } \mathrm{d}z(h_\varepsilon - z )
  \big( \overline{P_\varepsilon^\perp  \psi_0}, \eta_\varepsilon(z)\big)+
 \int_{\mathcal{C}_r} \mathrm{d}z \, z (
\overline{P_\varepsilon^\perp  \psi_0} , \eta _\varepsilon (z) )_{L^2 (\R^d )}
\\ \nonumber
= \ & \frac{i}{2\pi }  (h_0 -
h_\varepsilon  )\big(  \overline{P_\varepsilon^\perp  \psi_0 }, \int_{\mathcal{C}_r
} \mathrm{d}z \, R_0 (z)\psi_0 \big)
\\  \nonumber
&
+\int_{\mathcal{C}_r } \mathrm{d}z \, z \,
\big(\overline{ P_\varepsilon^\perp  \psi_0}, \eta _\varepsilon (z)  \big)_{L^2 (\R^d )}
\\
\label{eq-aux1}
= \ & (h_0 - h_\varepsilon  )
\big(
\overline{P_\varepsilon^\perp  \psi_0}, \psi_0 \big) +
 \int_{\mathcal{C}_r } \mathrm{d}z \, z \,
  \big( \overline{P_\varepsilon^\perp  \psi_0} , \eta _\varepsilon (z)
  \big)_{L^2 (\R^d )}\,,
\end{align}
where we have  used the
fact $\frac{i}{2\pi } \int_{\mathcal{C}_r} \mathrm{d}z  \, R_0
(z)\psi_0 = \psi_0 $ and
\begin{eqnarray}
  (h_\varepsilon -z )\left( u,\big(R_\varepsilon (z)-R_0 (z)\big)\psi_0\right)
  =(h_0 - h_\varepsilon )(  u,  R_0 (z)\psi_0 )
\end{eqnarray}
valid for all $u\in H^{1} (\R^d)$ (\cf~\cite[Sec.~VIII.3.2]{Kato}).
It follows from (\ref{eq-P}) that
$$
  P_\varepsilon  \psi_0
  = ( \overline{\psi_\varepsilon} , \psi_0 )_{L^2 (\R^d)} \psi_\varepsilon
  = (1+O(\varepsilon))\psi_\varepsilon
  \,.
$$
Moreover,
\begin{equation}\label{eq-error}
P^\perp _\varepsilon  \psi_0 =(1+O(\varepsilon))\psi_\varepsilon  - \psi_0  \,,\,\,\,\,\,\,\,\,\, \| P_\varepsilon^\perp  \psi_0 \|_{L^2 (\R^d )} = O(\varepsilon )\,,
\end{equation}
which implies
\begin{equation}\label{eq-cloP1}
  \big(\overline{P_\varepsilon  \psi_0}, P_\varepsilon  \psi_0\big)_{L^2 (\R^d)}
  = 1+O(\varepsilon^2)
  \,.
\end{equation}
Using the above asymptotics,
we conclude that the second term
on the last line of~(\ref{eq-aux1})
behaves as $O(\varepsilon^2)$.

It remains to estimate the first term
on the last line of~(\ref{eq-aux1}).
Applying the notations $v_0 := \mathcal{U} \psi_0 $
and $w_\varepsilon :=  \mathcal{U}P_\varepsilon^\perp  \psi_0 $,
we  get
\begin{align} \nonumber \label{eq-aux2}
(  h_\varepsilon - h_0   )
 \big(
 \overline{P_\varepsilon^\perp  \psi_0 },  \psi_0  \big)
  = \ &
  \alpha_+ \underbrace{\int_{\Sigma_0 } \big( ( w_\varepsilon
 v_0 ) (q,  \varepsilon )
-  (w_\varepsilon
 v_0)( q , 0)\big) \, \mathrm{d}\Sigma_0}_{L^+ _1}
\\ \nonumber
 & +    \alpha_- \underbrace{\int_{\Sigma_0 }
\big( ( w_{\varepsilon}
 v_0 ) (q,  -\varepsilon )
-  (w_{\varepsilon}
 v_0 )( q , 0)\big) \, \mathrm{d}\Sigma_0}_{L^-_1}
\\  \nonumber
& + \alpha_+  \underbrace{\int_{\Sigma_0 }  ( w_\varepsilon
 v_0 ) (q, \varepsilon ) \, \big(f(q,  \varepsilon)-1 \big)
  \, \mathrm{d}\Sigma_0 }_{L^+_2}
  \\  \nonumber
&
 +\alpha_- \underbrace{\int_{\Sigma_0 }  ( w_{\varepsilon}
 v_0 ) (q, -\varepsilon )\big(f(q, - \varepsilon)-1 \big)
  \, \mathrm{d}\Sigma_0 }_{L^-_2}\,.
\end{align}
Using again the bound~(\ref{eq-supnorm})
together with the Schwarz inequality, we estimate
\begin{equation}\label{eq-I}
|L_2^\pm| \leq C \varepsilon \,
  \|w_\varepsilon \|_{L^2 (\Sigma_{\pm  \varepsilon})}
\|v_0 \|_{L^2 (\Sigma_{\pm  \varepsilon})}\,.
\end{equation}
Employing now the  statement of~Theorem~\ref{Thm.uniform}
we conclude
\begin{equation}\label{eq-weps}
  \|w_\varepsilon \|_{L^2 (\Sigma_{\pm  \varepsilon})} = O (\varepsilon)
  \,,
\end{equation}
which leads to $L_2^\pm =  O(\varepsilon^2 )$
in view of (\ref{eq-I}) and the fact that $\|v_0 \|_{L^2 (\Sigma_{\pm  \varepsilon})}$ can be uniformly bounded. This means that $L_2^\pm $ contributes to the error term.

To estimate $L_1^\pm $ we rely on the regularity of eigenfunctions
established in Lemma~\ref{Lem.regularity}.
For $g \in \{ w_\varepsilon , v_0\}$, we have
the expansion
$$
g (q,0 )= g (q, \pm \varepsilon ) \mp \varepsilon \, \partial_t g (q, \pm \varepsilon^\mp )+\breve{g}_\varepsilon \,,
$$
where $\breve{g}_\varepsilon\in L^2 (\Sigma_0 )$
admits the norm asymptotics of type $O(\varepsilon^2)$.
This implies
$$
  L_1^\pm = \pm \varepsilon L^\pm _3 \pm \varepsilon L^\pm _4
$$
with
$$
\begin{aligned}
  L^\pm _3 & :=
  \int_{\Sigma_0 }( \partial_t w_\varepsilon (q, \pm \varepsilon^\mp ))
  \, v_0 (q, \pm \varepsilon) \, \mathrm{d}\Sigma_0
  \,.
  \\
  L^\pm _4 & := \int_{\Sigma_0 } w_\varepsilon (q, \pm \varepsilon)
  \, \partial_t v_0 (q, \pm \varepsilon) \, \mathrm{d}\Sigma_0
  \,.
\end{aligned}
$$
Note that since~$v_0 $ is smooth for $t\neq 0$,
we do not need to distinguish ``left'' and ``right''
limits for $  \partial_t v_0 (q, \pm \varepsilon)$.
Employing  again (\ref{eq-weps})
and $\|\partial_t   v_0 \|_{L^2 (\Sigma_{\pm \varepsilon}\,)}\leq C$,
we claim that $L_4^\pm =O(\eps) $,
\ie\  $\varepsilon L_4^\pm $ contributes to  the error term.
It remains to estimate $L_3 ^\pm $.
To  this aim
we use the boundary conditions which for $v_\varepsilon$ read
$$
\partial_t v_\varepsilon (q, \pm \varepsilon^+  ) - \partial_t v_\varepsilon (q, \pm \varepsilon^- )
=  \alpha_\pm v_\varepsilon (q, \pm \varepsilon )\,.
$$
Using these
equivalences and  decomposition~(\ref{eq-error}),
we obtain
\begin{multline}
 L_3^\pm = - \alpha_\pm \int_{\Sigma_0 }( v_\varepsilon v_0 )(q, \pm \epsilon )\, \mathrm{d}\Sigma_0\\ \nonumber
  \pm
 \int_{\Sigma_0 } \partial_t \big(
  v_\varepsilon (q, \pm \varepsilon^\pm ) -
v_0 (q, \pm \varepsilon^\mp )\big) v_0 (q, \pm \varepsilon)\,
 \mathrm{d}\Sigma_0 +O(\varepsilon)\,.
\end{multline}
Employing again $v_0 (q, \pm \varepsilon^\mp ) = v_0 (q, \pm \varepsilon^\pm ) $
and combining it with the statements of Theorem~\ref{Thm.uniform}
and Lemma~\ref{Lem.derivative},
we obtain
$$
  L_3^\pm = -  \alpha_\pm \int_{\Sigma_0 } v_0 ^2 \, \mathrm{d}\Sigma_0
  + O(\eps )\,.
$$
Summing up, the above estimates we come to~(\ref{eq-estim2.bis}),
which completes the proof.
\end{proof}
Now we are in a position to establish Theorem~\ref{Thm.evs}.
\begin{proof}[Proof of Theorem~\ref{Thm.evs}]
Combining (\ref{eq-P}),  (\ref{eq-estim2.bis}),
(\ref{eq-cloP1}) and (\ref{eq-formev}) we get
$$
\lambda_{\varepsilon} =
\frac{h_\varepsilon \big(\overline{P_\varepsilon
\psi_0 }, P_\varepsilon
\psi_0 \big)}
{\big(\overline{P_\varepsilon
\psi_0}, P_\varepsilon
\psi_0  \big)_{L^2 (\R^d )}}=
\lambda_0 + \varepsilon \lambda_0'+
O(\varepsilon^2)\,,
$$
where $\lambda_0'$ is defined by (\ref{as.bis}).
\end{proof}
\subsection{Degenerate eigenvalues}
In this subsection, we extend Theorem~\ref{Thm.evs}
to the case of degenerate eigenvalues.
More specifically, now we assume that~$\lambda_0$
is a discrete \emph{semisimple} eigenvalue of~$H_0$.
The semisimple property means that the algebraic multiplicity
can be greater than one, but it is still equal to
the geometric multiplicity of the eigenvalue
(\cf~\cite[Sec.~I.5.3]{Kato}).
It is the most general situation in the self-adjoint setting
(\ie\ $\alpha_\pm \in \Real$ in our case).

Let $k\in \N$ stand for the multiplicity of~$\lambda_0$
and let $\{\psi_0^i\}_{i=1}^{k}$ denote a system
of linearly independent eigenvectors of~$H_0$,
normalised in such a way that
the biorthonormal relations
\begin{equation}\label{basis}
  \big(\overline{\psi_0^i},\psi_0^j\big)_{\sii(\R^d)} = \delta_{ij}
\end{equation}
hold true for all $i,j \in \{1,\dots,k\}$.
We note that $\big\{\overline{\psi_0^i}\big\}_{i=1}^{k}$ constitutes a system
of linearly independent eigenvectors of the adjoint~$H_0^*$
corresponding to the semisimple eigenvalue~$\overline{\lambda_0}$
of the same multiplicity~$k$.

Our main result reads as follows.
\begin{theorem}\label{th-deg}
Let~$\lambda_0$ be a semisimple discrete eigenvalue of~$H_0$
of  multiplicity $k \geq 1$
and let~$\{\psi_0^i\}_{i=1}^k$ stand for
a system of the corresponding eigenfunctions
normalised via~\eqref{basis}.
There exist positive constants~$\eps_0$ and~$r$ such that,
for all $\eps < \eps_0$, $H_\eps$~possesses precisely~$k$
(counting the algebraic multiplicity)
discrete eigenvalues
$\{\lambda^i_\varepsilon \}_{i=1}^k$
in the open
disk of radius~$r$ centred at~$\lambda_0$.
Moreover, $\{\lambda^i_\varepsilon\}_{i=1}^k$ admit the following asymptotics
\begin{equation}\label{eq-evasHd}
   \lambda^i_\varepsilon = \lambda_0 \,
  + \lambda_{i}' \,\varepsilon
  + o(\varepsilon) \,
  \,,
\end{equation}
where $\{\lambda_i'\}_{i=1}^k$ are eigenvalues
(counting the algebraic multiplicity)
of the matrix $S \equiv \{s_{ij}\}_{i,j=1}^k$ with entries
\begin{multline}\label{eq-evmatrixHd}
s_{ij} :=
  \alpha_+ \int_{\Sigma_0} \partial_n^+ (\psi_0^i \psi_0^j)
  + \alpha_-   \int_{\Sigma_0} \partial_n^- (\psi_0^i \psi_0^j)
  \\
  - \int_{\Sigma_0}
  \left[\alpha_+^2+\alpha_-^2 + (\alpha_+ -\alpha_-) \,(d-1) K_1 \right]
  \psi_0^i \psi_0^j
  \,.
\end{multline}
\end{theorem}
\begin{proof}
Relying again on the norm-resolvent convergence of Theorem~\ref{Thm.nrs},
we can choose $r>0$ in such a way that the circle~$\mathcal{C}_r$
introduced in~\eqref{circle} surrounds~$k$ eigenvalues
of~$H_\varepsilon $ for all~$\varepsilon$ small enough.
These eigenvalues admit the following asymptotics
\begin{equation}\label{eq-as-ev-hD}
  \lambda^i_\varepsilon = \lambda_0 +O(\varepsilon )\,,\qquad i=1,\dots, k
  \,.
\end{equation}
Let us denote by
$\psi_\eps^i$, $i=1,...,k$,
the corresponding linearly independent  eigenfunctions of~$H_\eps$
with the normalisation
$(\overline{\psi_\eps^i} , \psi_\eps^i)_{L^2 (\R^d)}=1$.
Then we can find a system $\{\psi_0 '^i\}_{i=1}^k$  of  eigenfunctions of $H_0$ corresponding to
$\lambda _0$ such that
\begin{equation}\label{eq-convev}
  \psi_\eps^i = P_\eps \psi  _0 '^i \,,
\end{equation}
where $P_\eps$ stands for the eigenprojector onto
the space spanned by $\{\psi _\eps ^i \}_{i=1}^k$.
To show (\ref{eq-convev}) it suffices to check that
$\{ P_\varepsilon \psi _0 '^i  \}_{i=1}^k$ forms a basis in $\Ran(P_\varepsilon)$.
Using the convergence~\eqref{eq-P} of spectral projections
defined by~\eqref{projector},
we get the asymptotics
\begin{equation}\label{eq-P3}
  \|P_\varepsilon \psi _0 '^i -\psi _0 '^i \|_{L^2 (\R^d)}
  = \|P_\varepsilon \psi _0 '^i -P_0\psi _0 '^i  \|_{L^2 (\R^d)}= O(\varepsilon)
\end{equation}
for $i=1,\dots, k$. Consequently,
\begin{align}\nonumber
\big(\overline{P_\varepsilon \psi _0 '^i } , P_\varepsilon \psi _0 '^j \big)_{L^2 (\R^d)}
&=
\big(\overline{\psi _0 '^i } ,  \psi _0 '^j \big)_{L^2 (\R^d)} -
\big(\overline{P_\varepsilon ^\bot \psi _0 '^i } , P_\varepsilon^\bot  \psi _0 '^j\big)_{L^2 (\R^d)}
\\ \label{eq-deg1}
& = \big(\overline{\psi _0 '^i} ,  \psi _0 '^j \big)_{L^2 (\R^d)} +O(\varepsilon^2)
\,.
\end{align}
 It follows from the above asymptotics that
$\{ P_\varepsilon \psi _0 '^i\}_{i=1}^k$ forms a linearly independent system.
Actually, $\{ P_\varepsilon \psi'^i_0 \}_{i=1}^k$ constitutes
a basis of the range of~$P_\varepsilon$,
since  $\dim\Ran(P_\varepsilon) = k$.

The eigenvalues $\lambda_\eps^i$ of $H_\eps$ are determined by the
eigenvalues of the diagonal matrix
$$
  D:= \{d_{i } \delta_{ij}\}_{i,j=1}^k
  \qquad \mbox{with} \qquad
  d_i :=
  \big( \overline{H_\eps \psi^i_\eps} ,\psi^i_\eps  \big)_{L^2 (\R^d)}
  = h_\varepsilon
  \big(\overline{P_\varepsilon \psi'^i_0 }, P_\varepsilon \psi'^i_0 \big )
  \,.
$$
Now we repeat the steps from the proof of~Theorem~\ref{Thm.uniform}
and show
$$
  \|\psi^i_\varepsilon -\psi'^i_0\|_{L^\infty (\Sigma_{\pm \eps })}
  =O(\varepsilon)
$$
for $i=1, \dots, k$.
Furthermore, we employ the decomposition
$$
  h_\varepsilon \big(\overline{P_\varepsilon \psi'^i_0 }, P_\varepsilon \psi'^i_0 \big)
  = h_0 \big( \overline{\psi'^i_0} , \psi'^i_0 \big)
  + (h_\varepsilon - h_0 ) \big( \overline{\psi '^i_0} , \psi'^i_0\big)
 - h_\varepsilon \big(\overline{P_\varepsilon^\perp  \psi'^i_0},
  P_\varepsilon^\perp  \psi'^i_0 \big)
  \,.
$$
Repeating the arguments from the proofs of
Proposition~\ref{prop-1} and  Lemma~\ref{le-as3},
we establish
\begin{multline}   
  \lim_{\eps \to 0}
  \frac{h_\varepsilon \big( \overline{\psi'^i_0} , \psi'^i_0 \big)
  - h_0 \big( \overline{\psi'^i_0} , \psi'^i_0 \big)}
  {\eps}
  \\ \label{eq-aux10}
  = \alpha_+ \int_{\Sigma_0} \partial^+_n (\psi'^i_0 \psi'^i_0 ) +\alpha_- \int_{\Sigma_0} \partial^-_n (\psi'^i_0 \psi'^i_0 )
 - (\alpha_+ -\alpha_-)(d-1) \int_{\Sigma_0 }  K_1 \,  \psi'^i_0 \psi'^i_0
 \end{multline}
and
\begin{equation}\label{eq-estim4}
 h_\varepsilon \big(\overline{P_\varepsilon^\perp  \psi'^i_0},
  P_\varepsilon^\perp  \psi'^i_0 \big)
 =  \varepsilon \, (\alpha_+ +\alpha_-) \int_{\Sigma_0  } \psi'^i_0 \psi'^i_0
 +O(\varepsilon^2 )
  \,.
\end{equation}
 Since $\{\psi'^j_0\}_{j=1}^k$ is a basis,
 we can express any vector $\psi^j_0 $, $j=1,\dots,k$
satisfying biorthonormal relation~(\ref{basis}),
 as a linear combination
  $ \psi^j_0= \sum_{i=1}^k a_{ji }\psi'^i_0$, where  $a_{ji} \in \C$.
Furthermore, let us define matrix $S'$ as $D$ expressed in the new basis, precisely
$$
S':=\{s'_{ij}\}_{i,j=1}^k \,,\qquad \mathrm{with} \quad s'_{ij} = (\overline{a_i} , D a_j)_{l^2_k}\,,
$$
where $a_i : =(a_{i1},\dots, a_{ik}) \in l^2_k$. The eigenvalues of $S'$ and $D$ coincide. Furthermore, applying (\ref{eq-aux10}) and (\ref{eq-estim4}), we conclude that $s'_{ij} = \lambda_0 \delta_{ij} + s_{ij}\eps +O(\eps^2)$ which implies the claim.
\end{proof}

\appendix

\section{Appendix: Colliding quantum dots}\label{Sec.warm}
%
In this appendix, we focus on the special situation
of two approaching point interactions on the real line.
The simplicity of the problem enables one to derive
more precise asymptotic formulae by a different method.
At the same time, the explicit solutions provide a valuable insight
into the origin of the individual components in the first-order
correction term.

\subsection{Eigenvalue asymptotics}
As a special case of~\eqref{form},
we consider the m-sectorial operator~$H_\varepsilon$
associated with the  form
$$
  h_\varepsilon [\psi] :=
  \int_{\R}| \psi'(x)|^2 \, \D x
  +\alpha_+ |\psi (\varepsilon)|^2
  +\alpha_- |\psi  (- \varepsilon )|^2
  \,,
  \qquad
 \Dom(h_\varepsilon) := H^{1}(\R)
  \,.
$$
Note that the functions from $H^{1}(\R)  $ are continuous and, in this case,
the images of the  trace maps are just
determined by function values $\psi (\pm \varepsilon )$.

For $\varepsilon =0$, the operator~$H_0$
defines a well known model:
one-point interaction in one dimension with the coupling constant  $\alpha_++\alpha_-$.
The spectrum of~$H_0$ consists of
the essential (in fact continuous) spectrum $[0,\infty)$
and, under the condition $\Re (\alpha_+ +\alpha_-) <0$,
one simple discrete eigenvalue
\begin{equation}\label{eq-ev1D}
  \lambda_0 := -\frac{(\alpha_++\alpha_-)^2}{4}
\end{equation}
associated with the eigenfunction
$$
  \psi_0(x) := C_0 \, f_0(x)
  \,, \qquad f_0(x) := \mathrm{e}^{(\alpha_++\alpha_-)|x|/2}
  \,.
$$
Here the complex constant~$C_0$ is chosen in such a way that
the standard normalisation condition for non-self-adjoint spectral problems
$
  \int_\R \psi_0^2 =1
$
holds.

The case of two point interactions in one dimension
corresponding to $\eps > 0$ is also studied in the literature,
at least in the self-adjoint case
(see~\cite[Chap.~II.2]{AGH} and~\cite{KK4}).
The semi-axis $[0,\infty)$ still constitutes the essential spectrum of~$H_\eps$
and possible eigenvalues $\lambda_\eps$ equal $-\kappa_\eps^2$,
where $\kappa_\eps$ are determined as positive solutions of the implicit equation
\begin{equation}\label{eq-speqONE}
  (\alpha_+ +2 \kappa )  (\alpha_- +2 \kappa ) - \alpha_+ \alpha_-
  \mathrm{e}^{-4\kappa \varepsilon} = 0
  \,.
\end{equation}
For $\varepsilon$ small
enough equation (\ref{eq-speqONE}) admits a unique solution
$\kappa_\varepsilon$ which behaves as
\begin{equation}\label{as.kappa}
  \kappa_\varepsilon = \frac{\alpha_+ +\alpha_- }{2}
  + \alpha_+\alpha_- \varepsilon +O(\varepsilon^2 )
\end{equation}
as $\eps \to 0$.
The following theorem summarises the above discussion.
\begin{theorem}[$d=1$]
Let $\Re (\alpha_+ +\alpha_-) <0$.
For $\varepsilon$ small enough operator $H_\varepsilon$
has a unique simple discrete eigenvalue which admits the following asymptotics
\begin{equation}\label{eq-evONEd}
\lambda_\varepsilon=  \lambda_0-
(\alpha_++\alpha_-)\alpha_+\alpha_- \,\varepsilon
+O(\varepsilon^2 )
\end{equation}
or, equivalently,
\begin{equation}\label{eq-evONEda}
  \lambda_\varepsilon=  \lambda_0
  + \left[ \alpha_+ \, {\psi_0^2}'(0^+) - \alpha_- \, {\psi_0^2}'(0^-)
  -(\alpha_+^2+\alpha_-^2)\, \psi_0^2(0) \right] \varepsilon
  + O(\varepsilon^2 )
  \,.
\end{equation}
\end{theorem}
\begin{proof}
The first formula is due to~\eqref{as.kappa},
while its equivalent form follows by identities
${\psi_0^2}'(0^\pm) = \mp (\alpha_++\alpha_-)^2/2$
and $\psi_0^2(0) = -(\alpha_++\alpha_-)/2$.
\end{proof}

Note that~\eqref{eq-evONEda} is a special case of the general formula~\eqref{as.bis}.

Modified subtitle.
\subsection{More insight into the first-order correction term}
The aim of this section is to discuss in more detail
the first-order correction for the two-point interaction model.
In particular, we would like to analyse the source of
the term $-(\alpha_+^2 +\alpha_-^2)\psi_0 (0)^2$.

The general solution of
the eigenvalue problem $H_\varepsilon f_\eps = \lambda_\eps f_\eps$
takes the form
\begin{equation}\label{eq-efHepsilon}
  f_\varepsilon (x)=
  \begin{cases}
    \mathrm{e}^{\kappa_\varepsilon x} & \text{for}  \quad x<-\varepsilon \,,  \\
    c_1  \mathrm{e}^{- \kappa_\varepsilon x} +c_2  \mathrm{e}^{\kappa_\varepsilon x}
     & \text{for}  \quad -\varepsilon <x<\varepsilon\,, \\
     c_3 \mathrm{e}^{-\kappa_\varepsilon x} & \text{for} \quad  x>\varepsilon  \, .
  \end{cases}
\end{equation}
Using the boundary conditions~(\ref{conditions}) at $x=\pm\varepsilon$,
we determine the constants
\begin{equation}\label{eq-const12}
  c_1 = -\frac{\alpha_- }{2 \kappa_\varepsilon }
  \,, \qquad
  c_2 = \frac{\alpha_- +2\kappa_\varepsilon }{2 \kappa_\varepsilon }
  \,, \qquad
  c_3 =\mathrm{e}^{2\kappa_\varepsilon  \varepsilon }
  +\frac{\alpha_-}{2\kappa_\varepsilon} (\mathrm{e}^{2\kappa_\varepsilon  \varepsilon }-
  \mathrm{e}^{-2\kappa_\varepsilon  \varepsilon })
  \,.
\end{equation}
Moreover, employing~(\ref{eq-evONEd}), we get
\begin{equation}\label{eq-const3}
  c_3 = 1+O(\varepsilon)
\end{equation}
as $\eps \to 0$.
Let $\psi_\varepsilon $ stand for the normalised eigenfunction of $H_\varepsilon$,
\ie
$
  \psi_\varepsilon := C_\varepsilon f_\varepsilon
$,
where the complex constant~$C_\eps$ is chosen in such a way that
$\int_{\R}\psi_\varepsilon^2 =1$.
Let~$P_\varepsilon $ denote the corresponding eigenprojector, \ie
$$
  P_\varepsilon g :=
  \left(\overline{\psi_\varepsilon }, g \right)_{L^2 (\R )} \, \psi_\varepsilon\,.
$$
The eigenvalue~$\lambda_\varepsilon$ of~$H_\varepsilon $
satisfies
\begin{equation}\label{eq-formev-0}
  \lambda_\varepsilon
  =\frac{h_\varepsilon
  \big(\overline{P_\varepsilon \psi_0},P_\varepsilon \psi_0\big)}
  {\big(\overline{P_\varepsilon \psi_0}, P_\varepsilon \psi_0\big)_{L^2(\R)}}
  \,,
\end{equation}
where
\begin{equation}\label{eq-aux1-0}
  h_\varepsilon\big(\overline{P_\varepsilon \psi_0} , P_\varepsilon \psi_0\big)
  = h_0 \big(\overline{\psi_0}, \psi_0 \big)
  +(h_\varepsilon - h_0) \big(\overline{\psi_0}, \psi_0\big)
  -h_\varepsilon
  \big(\overline{P_\varepsilon ^\bot  \psi_0} , P_\varepsilon ^\bot \psi_0 \big)
\end{equation}
with  $P ^\bot :=I-P_\varepsilon$.
The first term on the right hand side of~(\ref{eq-aux1-0}) yields
$h_0\big(\overline{\psi_0 }, \psi_0\big) = \lambda_0$
since $\int_{\R} \psi_0^2 =1$.
The second term admits the asymptotics
\begin{eqnarray} \label{eq-estim1}
(h_\varepsilon - h_0 )\big(\overline{\psi_0}, \psi_0 \big)
  = \left( \alpha_+ {\psi_0^2}'(0+)
  - \alpha_- {\psi_0^2}'(0-)
  \right) \varepsilon +
O(\varepsilon^2)
  \,,
\end{eqnarray}
which reproduces the first two components
of the correction term in~(\ref{eq-evONEda}).

The remaining discussion is devoted to the analysis of
the third term on the right hand side of~(\ref{eq-aux1-0}).
A straightforward calculation
using (\ref{eq-const12}) and (\ref{eq-const3}) yields
\begin{equation}\label{eq-cal1}
  \left|\int_\R f^2_\varepsilon  - \int_\R f^2_0   \right|=O(\varepsilon )
  \,, \qquad
  \| f_\varepsilon - f_0\|_{L^2 (\R)} = O(\varepsilon)
  \,.
\end{equation}
Define
\begin{equation}\label{eq-estim2}
  \omega_\varepsilon  := P_\varepsilon ^\bot \psi_0 =
  \psi_0 - \big(\overline{\psi_\varepsilon}, \psi_0 \big)_{L^2 (\R)} \psi_\varepsilon
  \,.
\end{equation}
Note that the derivative of~$\omega_\varepsilon$ is well defined everywhere
apart $x=0$ and $x=\pm \varepsilon$.
Let $\omega_\varepsilon'$ denote this derivative.
Consequently, the third term on the right hand side
of~(\ref{eq-aux1-0}) takes the form
\begin{equation}\label{eq-estim20}
h_\varepsilon (\overline{\omega_\varepsilon}, \omega_\varepsilon ) = \int_{\R}
\omega_\varepsilon '^2  +\alpha_+ \, \omega_\varepsilon ^2(\varepsilon )
+\alpha_- \, \omega_\varepsilon ^2 (- \varepsilon )
  \,.
\end{equation}
%
Using again (\ref{eq-const12}) and (\ref{eq-const3}),
we state that
$$
  \omega_\varepsilon (\pm \varepsilon) =O(\varepsilon)
  \,.
$$
This means that the last two
terms on the right hand side~(\ref{eq-estim20})
behave as $O (\varepsilon^2)$.

Finally, let us analyse the first component~(\ref{eq-estim20}).
In view of~(\ref{eq-efHepsilon}), we decompose
$$
\int_{\R} \omega_\varepsilon '^2 = \int_{-\varepsilon}^\varepsilon
\omega_\varepsilon '^2  + \int_{-\infty }^{-\varepsilon}
\omega_\varepsilon '^2 +\int_{\varepsilon }^\infty
\omega_\varepsilon '^2 \,.
$$
A straightforward calculation shows that the last two terms
on the right hand side
behave as $O(\varepsilon^2)$. The first term requires a more
detailed analysis. Namely, for $x\in (0\,,\varepsilon)$ we have
$$
\omega_\varepsilon '(x)=- \frac{\alpha_++\alpha_-}{2}\left(-\kappa_0 \mathrm{e}^{-\kappa_0 x}+ c_1  \kappa_\varepsilon
\mathrm{e}^{-\kappa_\varepsilon x} -c_2\kappa_\varepsilon  \mathrm{e}^{\kappa_\varepsilon x}\right)^2+O(\varepsilon )= -\alpha_+ +O(\varepsilon)\,,
$$
where we have used (\ref{eq-const12}) together with the fact $\int_{\R}f^2_0 =-\frac{2}{\alpha_++\alpha_-}$. Analogously we show
$\omega_\varepsilon '(x)=-\alpha_- +O(\varepsilon)$ for $x\in (-\varepsilon\,,0)$.
This implies $\int_{-\varepsilon}^\varepsilon \omega_\varepsilon '^2
=(\alpha_+^2 +\alpha_-^2)\varepsilon +O(\varepsilon^2)$, and consequently,
$$
\int_{\R} \omega_\varepsilon '^2  = (\alpha_+^2 +\alpha_-^2) \varepsilon +
O(\varepsilon^2)\,,
$$
which, finally, leads to
$$
h_\varepsilon (\overline{\omega_\varepsilon}, \omega_\varepsilon )  = (\alpha_+^2 +\alpha_-^2) \varepsilon +
O(\varepsilon^2)\,.
$$
On the other hand,
$$
  (\overline{P_\varepsilon \psi_0}, P_\varepsilon \psi_0)_{L^2(\R)} =
(\overline{\psi_0}, \psi_0 )_{L^2 (\R)}-  (\overline{\omega_\varepsilon},\omega_\varepsilon )_{L^2 (\R)} = 1 +O(\varepsilon^2)\,.
$$

Summing up the above discussion,
we have  obtained the total first-order
correction term in~(\ref{eq-evONEda})
and identified the origin of its individual terms.

\subsection*{Acknowledgements}
The work was  supported
by the project RVO61389005 and the GACR grant No.\ 14-06818S as well as
 by the project DEC-2013/11/B/ST1/03067 of the Polish National Science Centre.
S.K. thanks the Department of Theoretical Physics, NPI CAS
in Re\v{z}, for the hospitality in July  2015, when some problems of this work were discussed.

%

\bibliography{bib}
\bibliographystyle{amsplain}

\end{document}